\documentclass[
  aps,
  prx,
  reprint,
  longbibliography,
  superscriptaddress
]{revtex4-2}
\usepackage{orcidlink}
\usepackage{template}

\begin{document}
\title{Trading Imaginary Time for Randomness in Ground State Preparation}

\author{Alvan Arulandu\,\orcidlink{0009-0000-0071-348X}}
\email{aarulandu@college.harvard.edu}

\affiliation{Department of Mathematics, Harvard University, Cambridge, MA 02138, USA}
\affiliation{Department of Computer Science, Harvard University, Cambridge, MA 02138, USA}
\affiliation{Harvard Quantum Initiative, Harvard University, Cambridge, MA 02138, USA}

\author{John M. Martyn\,\orcidlink{0000-0002-4065-6974}} 
\affiliation{Physical and Computational Sciences, Pacific Northwest National Laboratory, Richland, WA 99354, USA}
\affiliation{Harvard Quantum Initiative, Harvard University, Cambridge, MA 02138, USA}
\affiliation{Department of Physics, Harvard University, Cambridge, MA 02138, USA}

\author{Isaac L. Chuang\,\orcidlink{0000-0001-7296-523X}}
\affiliation{Department of Physics, Massachusetts Institute of Technology, Cambridge, MA 02139, USA}
\affiliation{Department of Electrical Engineering and Computer Science,
Massachusetts Institute of Technology, Cambridge, MA 02139, USA}

\begin{abstract}
    Imaginary-time evolution (ITE) is a foundational method for ground state preparation on quantum computers. However, because ITE is non-unitary, existing implementations incur a sample complexity and/or classical cost that scales exponentially with the target imaginary time $\beta$. 
    Moreover, the state itself converges slower than the energy, making accurate estimation of arbitrary ground state observables even more expensive. 
    In this work, we improve upon standard ITE by introducing \emph{twirled imaginary-time evolution} (TITE), which pairs ITE with real-time evolution applied for a random duration drawn from a carefully designed distribution. 
    We prove that this randomization quadratically suppresses the trace distance to the ground state, and thus also the error of arbitrary observables, which allows roughly half of the imaginary time to be replaced with real-time evolution ($\beta \mapsto \beta/2$) while maintaining the same level of accuracy. Because real-time evolution is unitary and does not incur an overhead in sample complexity or classical computation, this affords a quadratic reduction in the cost of any black-box ITE implementation, including Trotterization and quantum imaginary-time evolution. We demonstrate the efficiency of our algorithm in noisy circuit-level simulations of a non-integrable Ising chain, showing substantial improvements over standard ITE. 

\end{abstract}

\maketitle

\section{Introduction}

One of the most promising applications of quantum computers is the simulation of complex quantum systems. This includes tasks such as modeling the time dynamics of spin chains~\cite{andersen2025thermalization, kim2023evidence}, simulating quantum chemistry~\cite{Cao_2019}, and exploring quantum field theories~\cite{preskill2018simulating, Jordan_2014}. A particularly important task among these is estimating the ground state of a many-body Hamiltonian, which constitutes a central problem in condensed matter physics~\cite{Lin2020Near, Dong_2022, Peng_2005, cerezo2022variational, Yu_2023}, electronic structure~\cite{peruzzo2014variational, OMalley_2016, Hempel_2018, Romero_2019}, and even certain classes of optimization problems~\cite{farhi2014quantum, jordan2025optimization}.

Given this rich interest, several quantum algorithms have been developed to prepare ground states and estimate their properties. 
Among these, a particularly notable approach is \textit{imaginary-time evolution} (ITE). Given a Hamiltonian $\cH$, this method applies the non-unitary operator $e^{-\beta \mathcal{H}}$ to an initial state for a target imaginary time $\beta$. Upon renormalizing the state, this procedure suppresses excited-state contributions and projects out the ground state as $\beta$ increases.

Various strategies have been designed to realize imaginary-time evolution as a quantum algorithm. Prominent instances include using post-selection on Trotter approximations of $e^{-\beta \mathcal{H}}$~\cite{xie2022probabilistic, Kosugi_2022, TUrro_2022, Liu_2021, leadbeater2024non}, reformulating imaginary-time dynamics as a state-dependent real-time dynamics~\cite{Motta_2019}, and invoking quantum signal processing to approximate the exponential of the Hamiltonian~\cite{Gily_n_2019, Silva_2023}. These developments have made ITE a candidate algorithm for near-term and early fault-tolerant quantum computers. 
This is in contrast to quantum phase estimation which typically requires many ancilla, controlled time evolution, and deep circuits, even when the initial state has high overlap with the ground state. Further, unlike adiabatic algorithms which require spectral gap bounds along the entire adiabatic path, the sample complexity of ITE can be bounded in terms of only the target accuracy and the spectral gap of the target Hamiltonian.

While ITE approaches have been successfully demonstrated on small systems~\cite{Sun_2021, Yeter_Aydeniz_2021}, a key challenge in their deployment is their cost: applying the non-unitary operator $e^{-\beta \mathcal{H}}$ decreases the norm of the underlying wave function, which must be corrected through repeated circuit executions or amplitude amplification. Consequently, imaginary-time evolution algorithms often acquire a cost (e.g., sample complexity, amplitude amplification steps, or classical cost) that scales exponentially with $\beta$, prohibiting extension to large values of $\beta$ required for accurate ground state estimation. For example, methods based on Trotterization and quantum signal processing~\cite{xie2022probabilistic, Kosugi_2022, TUrro_2022, Liu_2021, leadbeater2024non, Gily_n_2019, Silva_2023} require post-selection, leading to a success probability that decays exponentially with $\beta$. In a similar vein, the ``quantum imaginary-time evolution'' algorithm developed in Ref.~\cite{Motta_2019} requires solving classical optimization problems on a Hilbert space whose dimension grows exponentially with $\beta$.

In this work, we present a general method to alleviate the cost of ITE algorithms by pairing imaginary-time evolution with real-time evolution. Specifically, we develop a method that replaces roughly half of the imaginary time $\beta$ with an equivalent amount of randomly-sampled real-time evolution. Because real-time evolution is unitary and preserves the norm of the wave function, it does not increase the classical cost or sample complexity, and consequently our method reduces the total cost to pinpoint the ground state. Specifically, because the sample complexity/classical cost generically scales exponentially in $\beta$, halving the amount of required imaginary time corresponds to an asymptotic quadratic reduction in cost.

Our method, which we term \textit{twirled imaginary-time evolution} (TITE), constructs a quantum channel composed of imaginary-time evolution for a time roughly $\beta/2$, followed by real-time evolution for a random time $t$ drawn from a carefully designed distribution, such that the average real-time dynamics reproduces the remaining $\beta/2$ of imaginary time. We refer to this randomized real-time evolution process as \textit{temporal twirling}.
We prove our results by using the concept of \textit{randomized compiling}, which provides a general framework for suppressing errors by mixing over unitaries. While this is traditionally used to mitigate gate errors in quantum circuits, here we use this tool to suppress the trace distance to the ground state in ITE. In addition, because our method is formulated at the level of a quantum channel, it is compatible with a broad range of existing ITE implementations, such as Trotterization, quantum signal processing, and quantum imaginary-time evolution. This flexibility allows our method to be readily integrated into existing and future algorithms for ground state estimation.

In the broader landscape of quantum information, our algorithm sits between error mitigation~\cite{Cai_2023} and ground state distillation methods (like virtual cooling and distillation~\cite{Cotler_2019, Huggins_2021}). Error mitigation aims to suppress physical noise in quantum states, often by leveraging tools like randomized compiling~\cite{Wallman_2016}. On the other hand, ground state distillation methods enhance purity by extracting cleaner ground states from noisy approximations, where this noise may stem from physical sources (e.g., gate error) or algorithmic sources (e.g., a low-temperature thermal state not being exactly the ground state). Our approach connects these two methods by randomly compiling over real-time evolution to mitigate the algorithmic error suffered by a state evolved for a finite imaginary time. Notably, this is accomplished without significantly increasing the cost relative to standard imaginary-time evolution.

We outline this work as follows. We begin by covering preliminaries and formulating our goal in Sec.~\ref{sec:prelim}, and then present our algorithm in Sec.~\ref{sec:algorithm}, where we also analyze its performance and connections to related topics. We then deploy our method in Sec.~\ref{sec:sim} to prepare the ground state of a non-integrable quantum spin chain, demonstrating its advantage over ordinary ITE in both noisy circuit-level and ideal simulations. Finally, we conclude and discuss the outlook of this work in Sec.~\ref{sec:conclusion}.

\section{Background}\label{sec:prelim} 

In this section, we provide background information on the two primary components of this paper, namely imaginary-time evolution and randomized compiling. Sec.~\ref{sec:background-ite} reviews imaginary-time evolution and its associated $\text{exp}(\cO(\beta))$ cost, and Sec.~\ref{sec:background-rc} reviews randomized compiling and the error suppression that it affords.

\subsection{Imaginary-Time Evolution}\label{sec:background-ite}

As its name suggests, imaginary-time evolution is the evolution of a quantum state upon replacing the real time $t$ with imaginary time $\beta = it$ in the Schr\"odinger equation. For a time-independent Hamiltonian $\cH$, ITE maps the time evolution operator $e^{-i\cH t}$ to the non-unitary ITE operator $e^{-\beta \mathcal{H}}$. This acts on an initial state $|\psi(0)\rangle$ as:
\begin{equation}
    \ket{\psi(\beta)} = \frac{e^{-\beta \cH}\ket{\psi(0)}}{\| e^{-\beta \cH}\ket{\psi(0)} \|_2} . 
\end{equation}

For sufficiently large $\beta$, the state $|\psi(\beta)\rangle$ converges to the ground state of $\mathcal{H}$, as excited-state contributions are exponentially suppressed by $\beta$. More precisely, let $\mathcal{H}|\lambda_k\rangle = \lambda_k |\lambda_k\rangle$ denote the eigenbasis of $\mathcal{H}$, and let the spectral gap be $\lambda_1 - \lambda_0 =: \Delta > 0$. Then, provided the initial state has non-zero overlap with the ground state, $|\langle \lambda_0 | \psi(0)\rangle| > 0$, the result of ITE is the ground state:
\begin{equation}
    \lim_{\beta \rightarrow \infty} \ket{\psi(\beta)} = |\lambda_0\rangle . 
\end{equation}
The rate of convergence depends on both the spectral gap and the target accuracy: to approximate the ground state to trace distance at most $\epsilon$ (equivalently, fidelity at least $1-O(\epsilon^2)$), it suffices to evolve for imaginary time $\beta = \cO( \log(1/\epsilon)/\Delta)$.

Given this guaranteed convergence to the ground state, ITE has formed the basis of many ground state preparation algorithms. Classically, it underlies powerful methods such as quantum Monte Carlo~\cite{Ceperley_1995, Foulkes_2001, Zhang_2003}, minimally entangled typical thermal states~\cite{White_2009, Stoudenmire_2010}, and time-evolving block decimation~\cite{Paeckel_2019, Kormos_2016}. More recently, ITE has inspired a variety of quantum algorithms for ground state preparation on near-term quantum computers~\cite{Motta_2019, xie2022probabilistic, Kosugi_2022, TUrro_2022, Liu_2021, leadbeater2024non, Silva_2023}. These algorithms aim to realize the action of the ITE operator through a quantum algorithm, thereby projecting an initial state onto the ground state. Here, the central challenge lies in simulating this non-unitary evolution using the inherently unitary dynamics provided by a quantum computer.

This challenge has prompted the development of several approaches for realizing ITE. Here we will consider three representative methods that implement the ITE operator directly (in contrast to variational quantum algorithms, which remain heuristic and highly sensitive to the choice of parameterization~\cite{McArdle_2019, Yuan2019theory}). Despite differences in their implementation and scope, these direct methods share a common yet costly limitation: their classical cost/sample complexity scales exponentially with the imaginary time, as $e^{\cO(\beta)}$.

First, a common approach is to approximate the ITE operator using Trotterization. Representing the Hamiltonian as a sum of local terms $\cH = \sum_{i=1}^L w_i h_i$, the simplest such implementation is the first-order Trotter formula over $r$ time steps: 
\begin{equation}
    e^{-\beta \cH}=\left(\prod_{i=1}^{L}e^{-w_i h_i \Delta \beta }\right)^{r}+\cO(\beta^2/r) ,  
\end{equation}
where $\Delta \beta = \beta / r$ is the imaginary-time step size. This decomposes the ITE operator into a sequence of local non-unitary operations $e^{-h_i \Delta \beta }$. Using this decomposition, Refs.~\cite{xie2022probabilistic, Kosugi_2022, TUrro_2022, Liu_2021, leadbeater2024non, ray_quasiprobabilistic_2025} designed algorithms that realize the local non-unitary operators probabilistically, either by post-selecting on mid-circuit measurements or using a quasiprobability decomposition. As a result, each time step only succeeds with sub-unity probability $p<1$, implying that the total success probability after $r = \beta / \Delta\beta$ time steps decays exponentially as $e^{-\cO(\beta)}$ (for fixed $\Delta \beta$). To ensure successful ground state preparation, one must compensate for this low success probability by performing $e^{\cO(\beta)}$ repetitions, thus costing exponentially many copies of the initial state. This remains true even if amplitude amplification is used to quadratically boost success probabilities.

Another similar example is the quantum imaginary-time evolution (QITE) algorithm of Ref.~\cite{Motta_2019}, which also uses the above Trotter decomposition. Instead of implementing each local non-unitary operation $e^{-w_i h_i \Delta \beta}$ probabilistically, the authors realize the action of each operation as a state-dependent unitary transformation that acts on a domain of $D$ qubits around the support of the term $h_i$. Crucially, this unitary depends on the current quantum state, and is determined by performing tomography and solving a classical problem on $D$ qubits. As noted in Ref.~\cite{Motta_2019}, the domain size $D$ grows with each step of the ITE algorithm, and thus scales linearly in $\beta$. This implies that both the tomography and classical optimization steps incur a sample cost and classical cost, respectively, of $e^{\cO(D)} = e^{\cO(\beta)}$, thus also incurring a complexity that grows exponentially with $\beta$.

Aside from Trotterization-based methods, recent works have also proposed realizing ITE with quantum signal processing (QSP)~\cite{Low_2016, Low_2017_Optimal}. QSP implements polynomial transformations of an operator, such that ITE can be achieved by designing a polynomial approximation to the decaying exponential: $P(\cH) \approx e^{-\beta \mathcal{H}}$, either directly~\cite{Gily_n_2019} or as a fragmented product of polynomials~\cite{Silva_2023}. In these constructions, the polynomial transformations are embedded in a sub-block of a larger unitary matrix and accessed by measurement and post-selection. Because the target polynomial $P(\cH) \approx e^{-\beta \mathcal{H}}$ decays exponentially in magnitude with $\beta$, so too does the probability of successfully accessing it. Consequently, like the probabilistic realizations of ITE, this low success probability must be addressed by performing $e^{\cO(\beta)}$ repetitions, even if amplitude amplification is used to boost success probability. 
Evidently, this $e^{\cO(\beta)}$ cost is a generic feature of leading implementations of ITE. 

\subsection{Randomized Compiling}\label{sec:background-rc}

Recent research in quantum information has highlighted the utility of randomization~\cite{arute2019quantum, zlokapa2023boundaries, schuster2025random, Ma_2025}. A notable such instance is the development of \textit{randomized compiling}. Rather than execute the same circuit at every call of an algorithm, randomized compiling proposes to execute a random circuit drawn from a distribution. Equivalently, this replaces a unitary operation with a quantum channel that is a probabilistic mixture of unitaries. In doing so, randomized compiling is able to suppress errors and achieve better performance than deterministic methods.

At the heart of randomized compiling is the \emph{mixing lemma}~\cite{Campbell_2017, hastings2016turning}, which quantifies the achievable error suppression:
\begin{lemma}[Campbell-Hastings Mixing Lemma~\cite{Campbell_2017, hastings2016turning}]\label{lemma:mixing}
    Let $V$ be a target unitary operator, and $\mathcal{V}(\rho ) = V \rho V^\dag$ the corresponding channel. Suppose there exist $m$ unitaries $\{U_j\}_{j=1}^m$ and an associated probability distribution $\{p_j\}$ that approximate $V$ as 
    \begin{equation}
    \begin{aligned}
        &\| U_j - V \|_{\rm op} \leq a \text{ for all } j, \\
        & \Big\| \sum_{j=1}^m p_j U_j - V \Big\|_{\rm op} \leq b
    \end{aligned}
    \end{equation}
    for some $a, b > 0$. Then, the corresponding channel $\Lambda(\rho) = \sum_{j=1}^m p_j U_j \rho U_j^\dag$ approximates $\mathcal{V}$ as 
    \begin{equation}
        \| \Lambda - \mathcal{V} \|_{\diamond} \leq a^2 + 2b . 
    \end{equation}
\end{lemma}

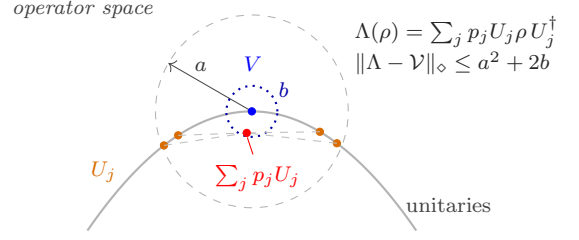
\begin{figure}[t]
\centering
\usetikzlibrary{calc}
\begin{tikzpicture}[scale=0.75, every node/.style={font=\small}]


\node[gray!50!black, anchor=west, font=\footnotesize\itshape] at (-4.4,4.0)
    {operator space};

\draw[gray!60, thick, domain=-2.9:2.9, smooth, samples=60]
    plot (\x, {2.2 - 0.25*\x*\x});
\node[gray!50!black, anchor=west, font=\footnotesize] at (2.55,0.55) {unitaries};

\draw[dashed, gray!60] (0,2.2) circle (1.7);
\draw[->, gray!60!black, thin] (0,2.2) -- ({1.7*cos(150)},{2.2+1.7*sin(150)});
\node[gray!40!black, font=\footnotesize] at (-0.90,3.00) {$a$};

\draw[dotted, thick, blue!65!black] (0,2.2) circle (0.45);
\node[blue!65!black, anchor=south west, font=\footnotesize] at (0.308,2.28) {$b$};

\foreach \x in {-1.55, -1.30, 1.20, 1.50} {
  \filldraw[orange!85!black] (\x,{2.2 - 0.25*\x*\x}) circle (1.8pt);
  \draw[gray!50, thin, dashed] (\x,{2.2 - 0.25*\x*\x}) -- (-0.09,1.82);
}
\node[orange!85!black, anchor=north east, font=\footnotesize] at (-2.22,1.48) {$U_j$};

\filldraw[red] (-0.09,1.82) circle (1.8pt);
\draw[red, thin] (-0.08,1.74) -- (0.02,1.44);
\node[red, anchor=north, font=\footnotesize] at (0.10,1.42) {$\sum_j p_j U_j$};

\filldraw[blue] (0,2.2) circle (1.8pt);
\node[blue, anchor=south, font=\footnotesize] at (0.02,2.74) {$V$};

\node[gray!30!black, anchor=west, font=\footnotesize] at (1.65,3.60)
    {$\Lambda(\rho)=\sum_j p_j U_j \rho\, U_j^\dagger$};
\node[gray!30!black, anchor=west, font=\footnotesize] at (1.65,3.05)
    {$\|\Lambda - \mathcal V\|_\diamond \leq a^2 + 2b$};

\end{tikzpicture}
\caption{The Campbell-Hastings mixing lemma (Lemma~\ref{lemma:mixing}). The unitaries $\{U_j\}$ lie on the manifold of unitaries (gray curve) within distance $a$ of the target $V$. Their weighted average $\sum_j p_j U_j$ leaves the manifold and lands within a smaller distance $b$ of $V$ because the first-order errors of $U_j - V$ point in different directions and cancel in the average. The induced channel $\Lambda(\rho) = \sum_j p_j U_j \rho U_j^\dag$ then deviates from the target channel $\mathcal V(\rho) = V\rho V^\dagger$ by only $a^2 + 2b$ in diamond norm.
}
\label{fig:mixing}
\end{figure}

To see how mixing suppresses error, consider an ensemble of unitaries $\{U_j\}_{j=1}^{m}$ each approximating a target unitary $V$ with error $a = \epsilon$, and a probability distribution $\{p_j\}$ such that the average unitary $\sum_j p_j U_j$ suffers error $b = \cO(\epsilon^2)$. The associated random channel $\Lambda(\rho) = \sum_{j=1}^m p_j U_j \rho U_j^\dag$ then deviates from the target channel by a smaller error $a^2 + 2b = \cO(\epsilon^2)$, which is quadratically smaller than the error of any $U_j$. Importantly, $\Lambda $ can be implemented by sampling $j\sim p_j$ and executing the corresponding unitary $U_j$, making it no more costly than executing any single approximation. Figure~\ref{fig:mixing} illustrates the geometry behind this suppression: the individual errors $U_j - V$ point in different directions and their first-order errors cancel on average, so the mixture is closer to the target unitary than any of its constituents.

Therefore, mixing over unitaries confers a quadratic suppression of error as $\epsilon \mapsto \cO(\epsilon^2)$, at only the additional cost of random sampling. In fact, this quadratic suppression is the best possible improvement achievable through randomized compiling; see Appendix~\ref{x:rc-quadratic-optimality} for a proof. 
Given this ability to suppress errors with randomization, recent works have used randomized compiling to accelerate real-time evolution~\cite{Campbell_2019_Random, childs_faster_2019, cho_doubling_2024, Ouyang_2020, kim2026randomizedproduct}, simplify the implementation of rotation gates~\cite{low2021halving}, reduce the cost of quantum signal processing~\cite{martyn_halving_2025}, and improve quantum control techniques~\cite{Yi2026, kim2026randomizedQ}. Here we will use it to reduce the cost of ITE by replacing a portion of imaginary time with randomized real-time evolution.

\section{\label{sec:algorithm}Twirled Imaginary-Time Evolution}

As randomized compiling has been used to improve a variety of quantum algorithms, it is natural to ask if it can also enhance ITE algorithms for ground state preparation. In this section, we provide an affirmative answer by introducing a randomized ITE algorithm with improved performance. Crucially, rather than apply randomized compiling to a specific implementation of ITE, we target the more general task of approximating the ground state from a state evolved in imaginary time, agnostic to the ITE implementation method.
As a result, we develop a randomized algorithm that suppresses errors while remaining compatible with a broad class of ITE implementations, including Trotterization, QITE, QSP, and beyond. Equivalently, this reduces the cost of approximating the ground state to a desired accuracy across this entire family of ITE methods. 

In developing this algorithm, we first motivate a trace-distance formulation of ground state preparation in Sec.~\ref{sec:alg-motivation}, and then present our algorithm and prove its error suppression in Sec.~\ref{sec:alg-main}. Next, we discuss its properties and implementation in Secs.~\ref{sec:alg-remarks} and~\ref{sec:alg-implementation}. We then optimize the probability distribution used in our algorithm in Sec.~\ref{sec:distrib-choice}, and lastly present connections of our algorithm to stabilizer states and the notion of ``structure versus randomness" in Sec.~\ref{sec:stabilizers}.

\subsection{Motivation}\label{sec:alg-motivation}
The key to our improvement is to note the following common, yet often overlooked, limitation of ITE algorithms for ground state estimation: achieving high accuracy in the ground state energy does not, in general, directly translate to a comparable accuracy in the expectation values of other observables. In fact, ITE produces a ground state approximation whose energy is significantly more accurate than other observables. This can lead to situations where the convergence of the energy might suggest that a sufficiently accurate approximation to the ground state has been achieved, while other relevant quantities, such as correlation functions, can remain significantly inaccurate.

To illustrate this point, consider using ITE to prepare an approximate ground state 
\begin{equation}
    \ket{\psi} = \sqrt{1-\delta^2}\ket{\lambda_0}+\delta\ket{\lambda_\perp} ,
\end{equation}
where $|\lambda_\perp\rangle$ is a linear combination of excited eigenstates, and $\delta \ll 1$. This state achieves an energy that differs from the true ground state energy by $\cO(\delta^2)$: 
\begin{equation}\label{eq:energy_error}
    \bra{\psi}\cH\ket{\psi}= \lambda_0+ \delta^2 \big( \langle \lambda_\perp | \mathcal{H} |\lambda_\perp \rangle -\lambda_0 \big) ,  
\end{equation}
which follows from $\bra{\lambda_0}\cH\ket{\lambda_{\perp}}=0$ by orthogonality of eigenstates. 
However, for an arbitrary observable $O$, where $\bra{\lambda_0} O \ket{\lambda_{\perp}} \neq 0$, the corresponding expectation value can deviate from the ground state value by an $\cO(\delta)$ error:
\begin{equation}
    \bra{\psi}O\ket{\psi}=\bra{\lambda_0}O\ket{\lambda_0}+2\delta\cdot \mathrm{Re}\{\bra{\lambda_0}O\ket{\lambda_\perp}\}+\cO(\delta^2) . 
\end{equation}

This discrepancy is captured by the trace distance between $|\psi\rangle$ and the ground state. Here and throughout, the trace distance between two states is defined as $\trd(\rho, \sigma) \equiv \frac{1}{2}\|\rho - \sigma\|_1$, where $\|A\|_1 = \Tr\sqrt{A^\dagger A}$ is the Schatten 1-norm (the sum of the singular values); for pure states it reduces to $ \sqrt{1 - |\langle \psi | \phi\rangle|^2}$. For the scenario above, the trace distance evaluates to:
\begin{equation}\label{eq:d_tr_psi_gs}
    d_{\text{tr}} \big( |\psi\rangle \langle \psi | , \ |\lambda_0 \rangle \langle \lambda_0 | \big) = \frac{1}{2} \big\| |\psi\rangle \langle \psi | - |\lambda_0 \rangle \langle \lambda_0 |  \big\|_1 = \delta,
\end{equation}
and upper bounds the error in an arbitrary observable as 
\begin{equation}
    \big| \langle \psi | O | \psi \rangle - \langle \lambda_0 | O | \lambda_0 \rangle \big| \leq 2 \| O \|_{\rm op} \cdot d_{\text{tr}} \big( |\psi\rangle \langle \psi | , \ |\lambda_0 \rangle \langle \lambda_0 | \big) . 
\end{equation}
In conjunction with Eq.~\eqref{eq:d_tr_psi_gs}, this implies that the error in an arbitrary observable is at most $\cO(\delta)$. We note that this is a worst-case bound, and certain observables, like the energy, can still achieve smaller errors. For generic observables though, the error is constrained to be $\cO(\delta)$.

Overall, this analysis indicates that using energy as a proxy for convergence can be misleading. ITE may appear to converge to the ground state energy, while suffering larger errors in other observables. A more reliable measure of convergence is the trace distance to the ground state. This suggests recasting the problem of ground state preparation with ITE as follows: 
\begin{problem}[Ground State Preparation via ITE]\label{prob:ground-qite}
    Let $\epsilon > 0$ be some small constant. Consider a Hamiltonian $\cH$ that has a gapped, non-degenerate ground state $|\lambda_0\rangle$, and suppose we have the ability to implement the imaginary-time evolution operator $e^{-\beta \cH}$ for any $\beta \geq 0$. Then, given an initial state $|\psi\rangle$ with non-zero ground state overlap $|\langle \psi | \lambda_0\rangle| > 0$, prepare a state $\rho$ such that
    \begin{equation*}
        \trd(\rho, \ketbra{\lambda_0}{\lambda_0})\leq \epsilon.
    \end{equation*}
\end{problem}


%
In general, solving Problem~\ref{prob:ground-qite} requires imaginary time $\beta = \cO(\log (1/\epsilon)/\Delta )$, and a corresponding sample complexity and/or classical cost $e^{\cO(\beta)}$. While the amount of imaginary time is modest, the sample complexity/classical cost grows exponentially with $\beta$, severely limiting the utility of ITE. Our goal is therefore to mitigate this cost by reducing the required imaginary time, while still achieving small trace distance error.

\subsection{Algorithm}\label{sec:alg-main}

Here we introduce \emph{twirled imaginary-time evolution} (TITE) as a randomized algorithm that more efficiently solves Problem~\ref{prob:ground-qite}. The central idea is to use randomized compiling to suppress contributions from high energy eigenstates, thus yielding a more accurate ground state approximation. We achieve this by applying random amounts of real-time evolution atop a state evolved in imaginary time --- that is, applying $e^{-i\cH t}$ for durations $t\sim \cD$ drawn from a probability distribution $\cD$. We refer to this randomization over real evolution times as \textit{temporal twirling}.

To build intuition for this, consider again the approximate ground state $\ket{\psi} = \sqrt{1-\delta^2}\ket{\lambda_0}+\delta\ket{\lambda_\perp} $ obtained by applying ITE to an initial state, where we again let $\delta$ denote the magnitude of the excited contribution, distinct from the final trace distance parameter $\epsilon$ in Problem~\ref{prob:ground-qite}. Applying real-time evolution $e^{-i \cH t}$ to this state imparts a relative phase $e^{-i(\lambda_k-\lambda_0)t}$ on the excited states comprising $|\lambda_\perp\rangle$. Thus, randomizing over the duration $t$ randomizes these phases, which washes out the excited-state contributions. As we will show below, this procedure allows us to quadratically suppress the trace distance error to $\mathcal{O}(\epsilon^2)$.

To make this statement rigorous, we begin by presenting a generalization of the Campbell-Hastings mixing lemma to quantum states, powered by the same error cancellation depicted in Figure~\ref{fig:mixing}:

\begin{lemma}[Mixing Lemma for States]\label{lemma:mixing-states-cont}
Let $\ket{v}$ be a target pure state. Suppose there exists a collection of pure states $\{ |u_j\rangle \}_j$ indexed by $j$ and an associated probability distribution $\cD$ over $j$ that obey 
\begin{equation}\label{eq:mixing-state-cont-condition}
    \begin{aligned}
        &\| |u_j\rangle - |v\rangle \|_2 \leq a  \text{ for all } j , \\
        & \big\| \mathbb{E}_{j \sim \mathcal{D}} \big[ |u_j\rangle \big] - |v\rangle \big\|_2 \leq b , 
    \end{aligned}
\end{equation}
for some $a,b > 0$, where $\| \cdot \|_2$ is the Euclidean 2-norm. Then, it follows that the corresponding mixed state $\rho = \mathbb{E}_{j \sim \mathcal{D}} \big[ |u_j\rangle \langle u_j | \big]$ deviates from $|v\rangle \langle v|$ in trace distance by
\begin{equation}
    d_{\text{tr}}\big( \rho ,  \ |v\rangle \langle v| \big) \leq \frac{a^2}{2} + b. 
\end{equation}
\end{lemma}

\begin{proof}
    Let $|w_j\rangle\equiv \ket{v}-\ket{u_j}$. Then, by triangle inequality, the Schatten 1-norm between $\rho$ and $|v\rangle\langle v|$ is upper bounded as
    \begin{align*}
        \biggr\lVert &\E_{j\sim\cD} \big[\ketbra{u_j}{u_j} \big] - \ketbra{v}{v}\biggr\rVert_{1}\\
        &=\biggr\lVert \E_{j\sim\cD} \Big[(\ket{v}-|w_j\rangle)(\bra{v}-\langle w_j |)\Big] - \ketbra{v}{v}\biggr\rVert_{1}\\
        &=\biggr\lVert \ket{v} \E_{j\sim\cD}[\bra{w_j} ]+\E_{j\sim\cD}[|w_j\rangle]\bra{v} -\E_{j\sim\cD}[|w_j\rangle\langle w_j | ]\biggr\rVert_{1}\\
        &\leq \left\lVert\left(\E_{j\sim\cD}\big[|w_j\rangle \big]\bra{v}\right)^\dagger\right\rVert_{1}
        +\left\lVert \E_{j\sim\cD}\big[|w_j\rangle \big]\bra{v}\right\rVert_{1}\\
        &\quad \ +\E_{j\sim\cD}\left[\lVert \ketbra{w_j}\rVert_{1} \right]\\
        &=2\left\lVert \E_{j\sim\cD}\big[|w_j\rangle \big]\bra{v}\right\rVert_{1}+\E_{j\sim\cD}\left[\lVert |w_j\rangle\langle w_j |\rVert_{1}\right] . 
    \end{align*}
    Because the 1-norm is the sum of the singular value magnitudes, it holds that $\lVert \ketbra{x}{y}\rVert_{1}\leq \lVert \ket{x} \rVert_2  \lVert \ket{y} \rVert_2 $ for unnormalized states $\ket{x},\ket{y}$. Applying this to the above inequality yields
    \begin{align*}
        &\leq 2\left\lVert \E_{j\sim\cD}[|w_j\rangle ]\right\rVert_{2}\lVert \bra{v}\rVert_{2}+\E_{j\sim\cD}[\lVert |w_j\rangle\rVert_{2}\lVert \langle w_j |\rVert_{2}]\\
        &\leq 2 \left\lVert \E_{j\sim\cD}[\ket{u_j}]-\ket{v}\right\rVert_{2}\cdot 1+\E_{j\sim\cD}[a^2 ]\\
        &\leq 2b+a^2 . 
    \end{align*}
    Recalling the definition $d_{\text{tr}}\big( \rho ,  \ |v\rangle \langle v| \big) = \frac{1}{2} \big\| \rho - |v\rangle \langle v| \big\|_1$ completes the proof.
\end{proof}

We will use this result to derive our twirled imaginary-time evolution method. In the language of Lemma~\ref{lemma:mixing-states-cont}, we will take the target pure state $\ket{v}$ to be exactly the ground state $\ket{\lambda_0}$. Then, we will find a collection of pure states that obey $a=\cO(\epsilon)$ and $b=\cO(\epsilon^2)$, such that by Lemma~\ref{lemma:mixing-states-cont} the corresponding mixed state solves Problem~\ref{prob:ground-qite} with $\cO(\epsilon^2)$ trace distance error. In more detail, these pure states are taken to be an initial state that is imaginary-time-evolved for duration $\beta$ and subsequently undergoes real-time evolution for a random duration $t_j$ sampled from a distribution $\cD$. The intuition is that the imaginary-time evolution guarantees that each pure state in the ensemble is $\cO(\epsilon)$-close to the ground state, while the random real-time evolution further suppresses the error to $\cO(\epsilon^2)$ for an appropriate distribution $\cD$ according to Lemma~\ref{lemma:mixing-states-cont}.


We now introduce twirled imaginary-time evolution:
\begin{theorem}[Twirled imaginary-time evolution]\label{thm:ground-qite-error-suppression}
Let $\ket{\psi}$ be some initial state with ground state overlap $\left|\bra{\psi}\ket {\lambda_0}\right|>c_{\rm min}>0$. Then, for some small constant $\epsilon >0$, set $\beta=\log(1/(c_{\rm min}\epsilon))/\Delta$, where $\Delta$ is a lower bound on the Hamiltonian's spectral gap: $\lambda_1-\lambda_0 \geq \Delta$. Now, suppose we apply imaginary-time evolution for imaginary time $\beta$ to the initial state to get the normalized state $\ket{\psi(\beta)}\equiv \frac{e^{-\beta \cH} \ket{\psi}}{\|e^{-\beta \cH}\ket{\psi}\|_2}$. Then, 
$$\trd(\ketbra{\psi(\beta)},\ketbra{\lambda_0})\leq \cO(\epsilon).$$
Next, let $\cD$ be some real-valued distribution such that $|\E_{t\sim \cD}[e^{-it\omega}]|\leq \cO(\epsilon)$ for all $\omega\geq \Delta $. Suppose we subsequently apply real-time evolution to $|\psi(\beta)\rangle$ for some random time $t\sim \cD$ to obtain an ensemble of states of the form $\ket{\psi(\beta,t)}\equiv e^{-i\cH t}\ket{\psi(\beta)}$ indexed by $t$. Then, the corresponding mixed state $\rho=\E_{t\sim \cD}[\ketbra{\psi(\beta,t)}]$ obeys
$$\trd(\rho,\ketbra{\lambda_0})\leq \cO(\epsilon^2),$$
which is quadratically closer to the ground state than the traditional imaginary-time-evolved state $\ket{\psi(\beta)}$.
\end{theorem}

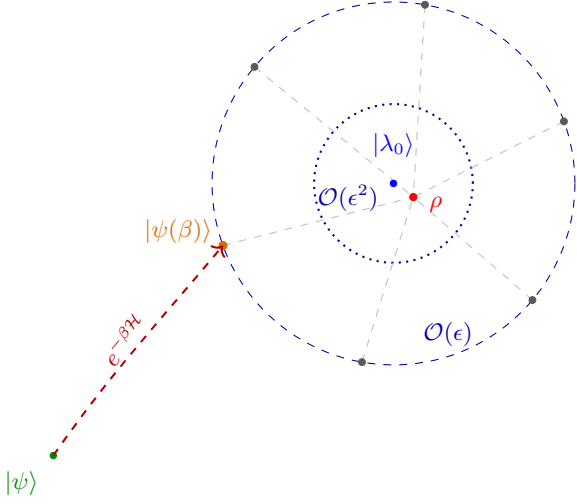
\begin{figure}[t]
\centering
\begin{tikzpicture}[scale=0.75, every node/.style={font=\small}]
\usetikzlibrary{calc}

\coordinate (gs) at (0,0);
\coordinate (rho) at (0.35,-0.24);

\draw[dashed, blue!65!black] (gs) circle (3.2);
\node[blue!65!black] at ($(gs)+(290:2.82)$) {$\cO(\epsilon)$};

\draw[dotted, thick, blue!65!black] (gs) circle (1.4);
\node[blue!65!black] at ($(gs)+(-0.80,-0.25)$) {$\cO(\epsilon^2)$};

\foreach \ang in {260,320,20,80,140} {
  \draw[gray!50, thin, dashed] ($(gs)+(\ang:3.2)$) -- (rho);
  \filldraw[gray!70!black] ($(gs)+(\ang:3.2)$) circle (1.7pt);
}
\draw[gray!50, thin, dashed] ($(gs)+(200:3.2)$) -- (rho);

\filldraw[orange!85!black] ($(gs)+(200:3.2)$) circle (2pt);
\node[anchor=south east, text=orange!85!black] at ($(gs)+(208:3.3)+(-0.12,0.33)$) {$\ket{\psi(\beta)}$};

\filldraw[red!100!black] (rho) circle (1.8pt);
\node[anchor=north west, text=red] at ($(rho)+(0.12,+0.179)$) {$\rho$};

\filldraw[blue] (gs) circle (1.6pt);
\node[anchor=south, text=blue] at ($(gs)+(0,0.30)$) {$\ket{\lambda_0}$};

\coordinate (psi) at (-6.0,-4.8);
\filldraw[green!55!black] (psi) circle (1.6pt);
\node[anchor=north east, text=green!55!black] at ($(psi)+(-0.1,-0.1)$) {$\ket{\psi}$};
\draw[->, thick, red!70!black, dashed] (psi) to[out=55, in=230]
  node[midway, above, sloped, font=\small, fill=white, inner sep=1pt] {$e^{-\beta \cH}$} ($(gs)+(200:3.2)$);

\end{tikzpicture}
\caption{Illustration of the twirled imaginary time evolution (TITE) algorithm (cf.\ Theorem~\ref{thm:ground-qite-error-suppression}). Imaginary time evolution maps the initial state $\ket{\psi}$ to $\ket{\psi(\beta)}$, within distance $\cO(\epsilon)$ of $\ket{\lambda_0}$. Applying real-time evolution for random durations $t\sim\cD$ produces an ensemble of states $\ket{\psi(\beta,t)}$ (gray points) spread around this same neighborhood, each differing from $\ket{\psi(\beta)}$ only by relative phases on its excited-state components. Because $\cD$ is chosen so that these phases average out, the ensemble-averaged state $\rho = \E_{t\sim \cD}[\ketbra{\psi(\beta,t)}{\psi(\beta,t)}]$ (red dot) lies within a much smaller distance $\cO(\epsilon^2)$ of the ground state.} 
\label{fig:algorithm}
\end{figure}

Figure~\ref{fig:algorithm} illustrates the intuition behind this result: imaginary-time evolution first brings the state to within $\cO(\epsilon)$ of the ground state, and then randomizing the subsequent real-time evolution averages away the relative phases on the excited-state components, suppressing the trace distance to $\cO(\epsilon^2)$. We can now rigorously prove the performance of TITE.

\begin{proof}
We begin by writing the initial state in the energy basis: 
$$\ket{\psi}=\sum_{k}c_k\ket{\lambda_k}$$ 
Then, applying imaginary-time evolution, up to global phase, 
\begin{equation}\label{eq:ite-state}
\begin{aligned}
        \ket{\psi(\beta)}&=\frac{c_0\ket{\lambda_0}+\sum_{k=1}c_ke^{-\beta(\lambda_k-\lambda_0)}\ket{\lambda_k}}{\sqrt{|c_0|^2+\xi^2}}\\
        &=\sqrt{1-\delta^2}\ket{\lambda_0}+\delta\ket{\lambda_{\perp}(\beta)}
\end{aligned}    
\end{equation}
where:
\begin{align*}
        \ket{\lambda_\perp(\beta)}&=\frac{c_0^*}{\xi|c_0|}\sum_{k=1}c_ke^{-\beta(\lambda_k-\lambda_0)}\ket{\lambda_k} , 
\end{align*}  
and $\delta\equiv \frac{\xi}{\sqrt{|c_0|^2+\xi^2}}$ and $\xi\equiv \sqrt{\sum_{k=1}|c_k|^2e^{-2\beta(\lambda_k-\lambda_0)}} \leq e^{-\beta \Delta}$. Because $\beta=\log(1/(c_{\rm min}\epsilon))/\Delta$ and $|c_0| > c_{\text{min}}$, it follows that $\delta\leq \frac{e^{-\beta \Delta}}{|c_0|}\leq \epsilon$. Then, 
$$\trd\big(\ketbra{\psi(\beta)},\ \ketbra{\lambda_0}\big)\leq \cO(\delta)=\cO(\epsilon),$$
which proves the baseline claim of Theorem~\ref{thm:ground-qite-error-suppression}. 

Now, consider the subsequent application of real-time evolution for some time $t$. 
\begin{align*}
    &\ket{\psi(\beta,t)}\\
    &=\sqrt{1-\delta^2} e^{-i \lambda_0 t}\ket{\lambda_0}+\frac{\delta}{\xi}\sum_{k=1} e^{-i\lambda_k t}c_k e^{-\beta(\lambda_k-\lambda_0)}\ket{\lambda_k }\\
    &=\sqrt{1-\delta^2}\ket{\lambda_0}+\delta \cdot \underbrace{\frac{1}{\xi}\sum_{k=1}e^{-i(\lambda_k-\lambda_0) t } c_k  e^{-\beta(\lambda_k-\lambda_0) }\ket{\lambda_k }}_{\ket{\lambda_\perp(\beta,t)}}
\end{align*}
where we have absorbed the global phase $e^{-i\lambda_0 t}$ arising from the ground state. We see that real-time evolution applies a phase of $e^{-i(\lambda_k-\lambda_0) t}$ to the $k$-th excited state contribution, where $\lambda_k-\lambda_0$ is the energy difference between the $k$-th excited state and the ground state, which is lower bounded by $\Delta$ for all $k$. The key idea is then that by randomizing the real-time evolution by drawing $t\sim \cD $, the magnitude of the expected phase factor is $\cO(\epsilon)$, because $|\E_{\cD}[e^{-it\omega}]|\leq \cO(\epsilon)$ for any $\omega\geq \Delta$. This is the heart of our quadratic error suppression, which we are then able to translate into an error suppression of trace distance via the mixing lemma for states.

Explicitly, we apply Lemma~\ref{lemma:mixing-states-cont} by taking the target pure state $\ket{v}$ to be the ground state $\ket{\lambda_0}$, and the ensemble $\{\ket{u_j}\}_{j\sim \cD}$ to be $\{\ket{\psi(\beta,t)}\}_{t\sim \cD}$. The first condition of Eq.~\eqref{eq:mixing-state-cont-condition} is satisfied with $a=\cO(\epsilon)$ because real-time evolution does not significantly degrade an approximation to the ground state. Formally, for any time $t$, the triangle inequality gives
\begin{align*}
&\|\ket{\lambda_0}-\ket{\psi(\beta,t)}\|_2\\
&=\|(1-\sqrt{1-\delta^2})\ket{\lambda_0}-\delta \ket{\lambda_\perp(\beta,t)}\|_2\\
&\leq \cO(\delta^2)+\cO(\delta) \leq \cO(\epsilon) . 
\end{align*}
Next, the second condition of \eqref{eq:mixing-state-cont-condition} is satisfied with $b=\cO(\epsilon^2)$ by the properties of $\cD $, which means that the average state of the ensemble achieves quadratically suppressed error to the ground state. Formally, 
\begin{align*}
    &\big\|\E_{t\sim \cD}[\ket{\psi(\beta,t)}]-\ket{\lambda_0}\big\|_2\\
    &=\bigg\|(\sqrt{1-\delta^2}-1)\ket{\lambda_0}\\
    &\hspace{1cm}+\delta\cdot \frac 1 \xi \sum_{k=1}\underbrace{\E_{t\sim \cD}[e^{-i(\lambda_k-\lambda_0)t}]}_{\cO(\epsilon)}c_k e^{-\beta(\lambda_k-\lambda_0)}\ket{\lambda_k }\bigg\|_2\\
    &\leq \cO(\delta^2)+\delta\cdot \cO(\epsilon)\cdot \|\ket{\lambda_\perp(\beta)}\|_2\\
    &=\cO(\epsilon^2)
\end{align*}
where we use the fact that $|\E_{t\sim \cD}[e^{-it\omega}]|\leq \cO(\epsilon)$ for all $\omega\geq \Delta$, from the stated assumptions on our distribution $\cD$. Thus, Lemma~\ref{lemma:mixing-states-cont} can be applied for $a=\cO(\epsilon)$ and $b=\cO(\epsilon^2)$ giving:
\begin{align*}
    \trd(\rho,\ketbra{\lambda_0}) = \frac{1}{2} \big\| \rho - \ketbra{\lambda_0} \big\|_1 \leq \frac {a^2}{2}+b=\cO(\epsilon^2), 
\end{align*}
which achieves our stated quadratic error suppression.
\end{proof}

Let us take a minute to analyze this result. Theorem~\ref{thm:ground-qite-error-suppression} indicates that the quadratic error suppression provided by TITE can be achieved using only additional real-time evolution, with no increase in imaginary-time evolution. Equivalently, because $\beta$ scales as $\cO(\log(1/\epsilon)/\Delta)$, the amount of imaginary time required to achieve a target error of $\epsilon$ can be halved ($\beta\to\beta/2$), at the expense of performing real-time evolution. Because real-time evolution is unitary and the sample complexity/classical cost of imaginary-time evolution algorithms scales as $e^{\cO(\beta)}$, TITE provides a quadratic improvement in the complexity of any black-box ITE algorithm.

\subsection{Remarks}\label{sec:alg-remarks}

It is important to observe that TITE does not decrease the energy because real-time evolution conserves energy, but instead improves estimation of all other observables by suppressing the trace distance to the ground state. From a thermodynamic perspective, TITE increases the entropy of the state (through randomization) to achieve this improvement. Moreover, in Appendix~\ref{x:rc-quadratic-optimality}, we show that this quadratic suppression of trace distance error is optimal, and that a generic super-quadratic suppression is impossible using the methods shown here. 

Furthermore, we note that the temporal twirling channel, namely $\E_{t\sim\cD }[e^{-i\cH t} \rho e^{i\cH t}]$, is capable of suppressing errors in approximate ground states, whether from imaginary-time evolution or other methods. Beyond ground states, this method also applies to any gapped excited state, so long as the conditions of Theorem~\ref{thm:ground-qite-error-suppression} hold for the excited state and its spectral gap.

Moreover, we mention that our algorithm is distinct from the random imaginary-time evolution algorithm developed in Ref.~\cite{Huang_2023}. There, the authors merge randomized Trotterization (specifically the QDrift algorithm~\cite{Campbell_2019_Random}) with quantum imaginary-time evolution~\cite{Motta_2019} to reduce its gate count. In contrast, our approach applies random real-time evolution, which may or may not employ Trotterization, on top of imaginary-time evolution to refine approximations to the ground state. In fact, because our approach is broadly compatible with realizations of imaginary-time evolution, it could be applied on top of the algorithm of Ref.~\cite{Huang_2023} to further enhance it.

Beyond ITE, our algorithm connects to further concepts in quantum information and computational physics. For one, the temporal twirling channel can be interpreted in the Heisenberg picture where it maps an observable as $ O\mapsto \E_{t \sim \cD}[e^{i\cH t} O e^{-i\cH t}]$. This map is known as the weighted operator Fourier transform, which has emerged as a key idea in recent work on Gibbs sampling \citep{jiang2026predicting, chen2023quantum}. This connection invites a broader investigation into ways one can leverage classical randomness to suppress errors in Gibbs sampling. Furthermore, it is worth noting a conceptual parallel between TITE and hybrid Monte Carlo (HMC) methods for path integration~\cite{Duane_1987}. There, the target ensemble is defined by an imaginary time (Euclidean) path integral weight $e^{-S_E}$, which is expensive to sample directly; HMC instead introduces fictitious momenta and evolves the system along cheap, energy-conserving trajectories to propose large, decorrelated moves, before correcting for integration error. Randomizing the trajectory length in HMC is furthermore known to be important for avoiding resonant, poorly-mixing behavior~\cite{Bou_Rabee_2017}. Our algorithm shares this basic division of labor: imaginary time defines and projects onto the target state, while cheap, unitary real-time evolution does the remaining work. 


\subsection{Implementation Considerations}\label{sec:alg-implementation}
Because Theorem~\ref{thm:ground-qite-error-suppression} is agnostic to the particular implementation of the imaginary-time evolution operator $e^{-\beta \cH}$, our TITE algorithm can be applied atop any existing ITE technique. However, careful consideration is still required when implementing TITE, because the choice of ITE implementation will affect the requisite gate count and circuit depth. Certain implementations may be better suited for particular platforms, depending on, e.g., their architecture, ability to support mid-circuit measurements, and noise resilience. 

Theorem~\ref{thm:ground-qite-error-suppression} also holds for any choice of real-time evolution algorithm. However, because these algorithms inevitably suffer systematic error (e.g., Trotter error), their accuracy must be sufficient to maintain the error suppression of Theorem~\ref{thm:ground-qite-error-suppression}. Straightforwardly, it can be seen that maintaining this error suppression requires a real-time evolution error $\cO(\epsilon^2)$, which does not violate the conditions of the mixing lemma for $a=\cO(\epsilon)$ and $b=\cO(\epsilon^2)$. 


Furthermore, while our TITE framework applies real-time evolution after imaginary time, this order can be relaxed. In fact, any interleaving of imaginary-time and real-time evolution steps is valid because $e^{i\cH t}$ commutes with $e^{-\beta \cH}$. In practice, however, these operations are implemented approximately, and interleaving them may induce errors that must be analyzed carefully. One advantage of performing all imaginary-time evolution first is that no real-time evolution effort is wasted if a post-selection step fails. On the other hand, it is possible that interleaving real-time and imaginary-time steps may offer improved noise resilience, suggesting a potential trade-off that warrants further investigation.

Let us also consider the real-time evolution overhead required by TITE. Suppose $C_{\rm RTE}(t,\epsilon)$ is the cost of implementing real-time evolution $e^{-i\cH t}$ to error at most $\epsilon$. For example, using first-order Trotterization, we have $C_{\rm RTE}(t,\epsilon) = \mathcal{O}(t^2/\epsilon)$. In TITE, the expected cost of real-time evolution is $C_{\rm RTE}(\epsilon)\equiv \E_{t\sim\cD}[C_{\rm RTE}(t,\epsilon)]$. For most practical real-time evolution algorithms, $C_{\rm RTE}(t,\epsilon)$ scales with $t$ as $\cO (t^\alpha)$ for an exponent $\alpha \in [1,2]$, so the expected overhead depends on the $\alpha$-th absolute moment of the distribution $\cD$, namely $\E_{t\sim\cD}[|t|^\alpha]$. In Section~\ref{sec:distrib-choice}, we construct distributions $\cD$ which minimize this cost.

\subsection{Optimizing the Distribution $\cD$}\label{sec:distrib-choice}
Theorem~\ref{thm:ground-qite-error-suppression} presupposes the existence of a distribution $\cD $ that obeys 
\begin{equation}\label{eq:D-phi-decay}
 |\E_{t\sim \cD}[e^{-i\omega t}]|\leq \cO(\epsilon)\quad \forall \omega\geq \Delta . 
\end{equation}
We can frame this condition in terms of the distribution's characteristic function $\varphi_{\cD}(\omega)\equiv \E_{t\sim \cD}[e^{-i\omega t}]$, which is the Fourier transform of the probability density: $|\varphi _{\cD}(\omega) | \leq \cO (\epsilon)$ for $\omega \geq \Delta$. However, it is not clear how to construct such a distribution a priori, nor how to minimize its corresponding real-time evolution cost $C_{\rm RTE}(\epsilon)$. As we discussed above, if the real-time evolution cost scales as $t^\alpha$, then $C_{\rm RTE}(\epsilon)$ depends on the $\alpha$-th absolute moment of the distribution $\cD$. Because $\alpha \in [1,2]$ in practice, this task is equivalent to minimizing the first and second absolute moments among the class of real-valued distributions with small tails in their characteristic functions. 

A straightforward way of satisfying Eq.~\eqref{eq:D-phi-decay} is to choose a distribution such that $\varphi_{\cD}(\omega)$ vanishes for all $\omega \geq \Delta$. We show in Appendix~\ref{x:distrib} Lemma~\ref{lemma:sinc-distrib} that the distribution $S_\Delta$ defined by the following simple probability density yields a characteristic function which vanishes outside of $[-\Delta,\Delta]$:
\begin{equation}
f_{S_\Delta}(t)\equiv \frac{3\Delta}{8\pi}\,{\rm sinc}^4(\Delta t/4).
\end{equation}
where ${\rm sinc}(x)=\sin(x)/x$ for $x\neq 0$ and ${\rm sinc}(0)=1$.
In fact, the absolute moments of this distribution scale as $\E_{\cD}[|t|^\alpha]=\cO(\Delta^{-\alpha})$, which we show is optimal up to constant factors in Appendix~\ref{x:distrib} Lemma~\ref{lemma:distrib-opt}. 
We use the notation $S_\Delta$ to emphasize that the distribution is constructed from the ${\rm sinc}(\cdot)$ function and only depends on the spectral gap lower bound $\Delta$. 

Of course, sampling $S_\Delta $ is impossible if $\Delta$ is unknown. If this is the case, we cannot simply choose some imaginary time $\beta=\Omega(\log(1/\epsilon)/\Delta)$ to achieve target error $\epsilon$. Instead, one typically chooses an educated guess for $\beta$, increasing its value until relevant observables noticeably converge. Now, for an arbitrary $\beta$, imaginary-time evolution gives a ground state preparation error of $\epsilon =\cO(e^{-\beta \Delta})$ by Eq.~\eqref{eq:ite-state}. However, TITE can still yield a quadratic error suppression to achieve a ground state error of $\cO(e^{-2\beta \Delta})$, so long as $|\varphi_{\cD}(\omega)|\leq \cO(e^{-\beta \Delta})$ for all $\omega \geq \Delta$. Since $\Delta$ is not known, any $\cD $ that is independent of $\Delta$ yet satisfies this condition must satisfy
\begin{equation}\label{eq:no-delta-decay}
    |\varphi_\cD (\omega)|\leq ce^{-\beta \omega}\quad \forall \omega>0
\end{equation}
for some constant $c>0$. 
We show in Appendix~\ref{x:distrib} Lemma~\ref{lemma:gcauchy} that the distribution $C_\beta$ defined by the following probability density satisfies these constraints for $c=\sqrt 2 $:
\begin{equation}\label{eq:f_C_beta}
    f_{C_\beta}(t)=\frac{1}{\beta\pi}\cdot \frac{1}{1+(t/\beta)^4/4} , 
\end{equation}
with moments $\E_{\cD}[|t|^\alpha]=\cO(\beta^\alpha)$ for $\alpha \in [1,2]$, which is optimal up to constant factors. We use the notation $C_\beta$ to emphasize that the distribution resembles the Cauchy distribution and only depends on $\beta$. 

Thus, $\cD=S_\Delta$ is the optimal choice of distribution when a bound $\Delta $ on the spectral gap is known, while $C_\beta$ should be used in the absence of such a bound. 
For a more familiar distribution, we also show in Appendix~\ref{x:distrib} Lemma~\ref{lemma:normal-distrib} that a normal distribution $\mathcal N_{\beta,\Delta}$ suffices to achieve quadratic error suppression, albeit suboptimally, if $\Delta$ is known. Figure~\ref{fig:distrib-schematic} summarizes these three constructions in terms of the decay condition of Eq.~\eqref{eq:D-phi-decay}.

\begin{figure}[t]
\centering
\begin{tikzpicture}[scale=1.15, every node/.style={font=\small}]
\usetikzlibrary{calc}

\def\xmax{6.6}
\def\ymax{3.35}
\def\Dx{4.0}   
\def\epsy{0.45} 
\def\topy{2.7}  

\fill[gray!25] (\Dx,\epsy) rectangle (\xmax,\ymax);
\node[align=center, text=gray!55!black, font=\scriptsize] at ($(\Dx,\epsy)+(1.2,0.55)$) {excluded by\\ Eq.~\eqref{eq:D-phi-decay}};

\draw[->] (0,0) -- (\xmax+0.2,0) node[right] {$\omega$};
\draw[->] (0,0) -- (0,\ymax) node[above] {$|\varphi_\cD(\omega)|$};

\draw[dashed] (\Dx,0) -- (\Dx,\ymax);
\node[below] at (\Dx,0) {$\Delta$};
\draw[dashed] (0,\epsy) -- (\xmax,\epsy);
\node[left] at (0,\epsy) {$\epsilon$};
\draw (0,\topy) -- (0.08,\topy) node[right=2pt, xshift=-18pt] {$1$};

\draw[very thick, green!55!black] plot [smooth] coordinates {
  (0,\topy) (0.6,2.55) (1.2,2.25) (1.8,1.8) (2.4,1.25) (2.9,0.75) (3.4,0.35) (3.8,0.08) (\Dx,0)
};
\draw[very thick, green!55!black] (\Dx,0) -- (\xmax,0);

\draw[very thick, blue] plot [smooth] coordinates {
  (0,\topy) (0.6,2.2) (1.2,1.75) (1.8,1.35) (2.4,1.0) (3.0,0.72) (3.6,0.52) (\Dx,0.42) (4.6,0.30) (5.2,0.20) (6.0,0.12)
};

\draw[very thick, orange] plot [smooth] coordinates {
  (0,\topy) (0.8,2.6) (1.6,2.35) (2.4,1.95) (3.0,1.55) (3.5,1.05) (3.8,0.65) (\Dx,0.44) (4.4,0.28) (5.0,0.16) (6.0,0.06)
};

\draw[very thick, orange] (0.2,3.05) -- (0.5,3.05);
\node[anchor=west] at (0.58,3.05) {$\mathcal N_{\beta,\Delta}$};
\draw[very thick, green!55!black] (1.95,3.05) -- (2.25,3.05);
\node[anchor=west] at (2.33,3.05) {$S_\Delta$};
\draw[very thick, blue] (2.95,3.05) -- (3.25,3.05);
\node[anchor=west] at (3.33,3.05) {$C_\beta$};

\end{tikzpicture}
\caption{Schematic comparison (not to scale) of the characteristic function magnitude $|\varphi_\cD(\omega)|=|\E_{t\sim \cD}[e^{-i\omega t}]|$ for the three distributions constructed in Sec.~\ref{sec:distrib-choice}. All satisfy $\varphi_\cD(0)=1$ and decay to at most $\cO(\epsilon)$ for $\omega\geq \Delta$ (shaded region), as required by Eq.~\eqref{eq:D-phi-decay}. $C_\beta$ and $\mathcal N_{\beta,\Delta}$ decay smoothly toward zero while $S_\Delta$ vanishes for $\omega\geq\Delta$.
}
\label{fig:distrib-schematic}
\end{figure}
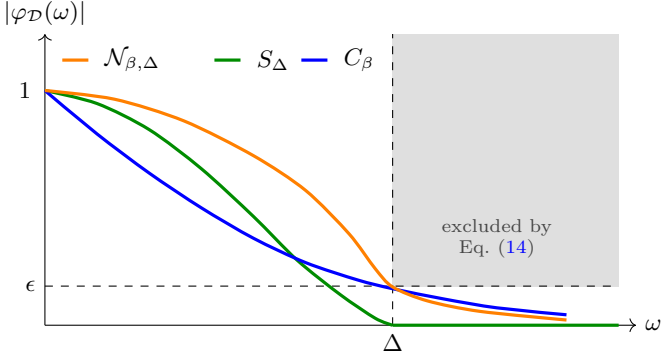

\subsection{Generalization to Stabilizer States and Beyond}\label{sec:stabilizers}

We have thus far presented how randomized real-time evolution enhances imaginary-time evolution. Beyond this specific application, our techniques connect with the broader concept of stabilizer states from quantum error correction.
To see this connection, let us assume without loss of generality that the ground state energy is $\lambda_0 = 0$, such that the ground state is invariant under real-time evolution. Equivalently, the ground state is the unique state \textit{stabilized} by the set of unitaries $\{ e^{-i \cH t} \}_{t\in \mathbb{R}}$: 
\begin{equation}
    e^{-i\mathcal{H}t} |\lambda_0\rangle = |\lambda_0\rangle  \ \  \forall t \in \mathbb{R} . 
\end{equation}
In TITE, we used this property to improve approximations of the ground state by randomly applying elements from the stabilizer set $\{ e^{-i \cH t} \}_{t\in \mathbb{R}}$.

This view suggests a broader generalization to stabilizer states. Suppose we have a set of unitary operators $\{ V_i \}$ that stabilize a target state $|\psi\rangle$ as $V_i |\psi\rangle = |\psi \rangle$ for all $i$. If we can only prepare an approximate state $|\tilde{\psi}\rangle$ that suffers trace distance error $d_{\text{tr}}(|\psi\rangle \langle \psi |, |\tilde{\psi} \rangle \langle \tilde{\psi} |) = \epsilon$, then one can improve this approximation by randomly applying the stabilizers $V_i$ to $| \tilde{\psi} \rangle$. Under suitable conditions on the sampling distribution (analogous to Lemma~\ref{lemma:mixing-states-cont}), this produces a state with smaller trace distance error $\cO(\epsilon^2)$, effectively distilling out a better approximation to $|\psi\rangle$.

This can be seen as a form of error mitigation applied to stabilizer states. Unlike standard stabilizer error correction, where stabilizers are measured to diagnose and correct errors, we apply stabilizers at random to suppress errors. There are, however, two key distinctions. First, this method offers error mitigation rather than full error correction: it achieves only quadratic error suppression and cannot be iterated indefinitely to arbitrarily reduce error. Second, we do not require that the stabilizing unitaries $V_i$ have eigenvalues $\pm1 $ as in typical stabilizer codes. Instead, they may have eigenvalues that are arbitrary complex phases, which broadens the scope of applicability (e.g., to unitaries generated by time evolution or other symmetry generators).

For completeness, we note that the use of randomization to improve stabilizer state preparation was previously explored in Ref.~\cite{greene2021error} under the name ``quantum measurement emulation''. Their approach seeks to combat gate noise, whereas here we use this idea to reduce algorithmic-level error in approximating a ground state. In addition, the construction of Ref.~\cite{greene2021error} assumes stabilizers with eigenvalues $\pm 1$, whereas we allow these to be arbitrary phases. This enables our formalism to target states stabilized by continuous symmetries, such as time evolution or translation invariance.

Beyond quantum information, the perspective developed in this section---namely, refining an approximation to a target state by randomizing over operations that stabilize it---admits a precise parallel with a foundational paradigm of modern additive combinatorics known as the \textit{structure versus randomness} dichotomy~\cite{tao2007structure}. In this paradigm, an arbitrary vector $f$ in a Hilbert space is decomposed as $f = f_{\rm str} + f_{\rm psd}$ (plus possibly a small residual error) with respect to a designated class of basic structured vectors. Here, $f_{\rm str}$ is a \textit{structured} vector, meaning that it is a controlled linear combination of basic structured vectors, while $f_{\rm psd}$ is a so-called \textit{pseudorandom} vector in the sense that its inner product with every basic structured vector is small in magnitude. (Note that the term ``pseudorandom'' here refers to its use in additive combinatorics~\cite{tao2007structure}, and should not be confused with its distinct meaning in quantum information).



Our TITE algorithm instantiates this dichotomy in the Hilbert space of the quantum system. The imaginary-time-evolved state $|\psi(\beta)\rangle = \sqrt{1-\delta^2}\ket{\lambda_0} + \delta \ket{\lambda_\perp(\beta)}$ is already a structure-plus-error decomposition, with the target $\ket{\lambda_0}$ playing the role of $f_{\rm str}$. The deficiency of plain ITE is that its error term is coherent and scales at first order in $\delta$, precisely because it is not pseudorandom. The randomized real-time evolution in TITE repairs this by rendering the error incoherent, scaling as $\cO(\delta^2) $, and thereby fitting naturally into the structure versus randomness decomposition. We expand on this connection in the context of the mixing lemma in Appendix~\ref{x:svr-mixing}.

\section{\label{sec:sim}Numerical Experiments}
Having proven the theoretical benefits of our TITE algorithm, we now benchmark it numerically. Our study proceeds in two stages. First, noisy circuit-level simulations establish the quadratic advantage under realistic conditions, including gate noise (Figure~\ref{fig:sim-noisy}). Second, noiseless simulations isolate the two systematic effects that can limit this advantage: Trotterization error (Figure~\ref{fig:sim-trotter-v-ideal}) and the choice of the twirling distribution $\cD$ (Figure~\ref{fig:sim-distrib}). We introduce the benchmark model in Sec.~\ref{sec:sim-model}, describe the circuit constructions and the two simulation regimes in Sec.~\ref{sec:sim-methods}, and present the results in Sec.~\ref{sec:sim-results}.

\subsection{Model and Setup}\label{sec:sim-model}
We consider preparing the ground state of the 1D mixed-field Ising model with open boundary conditions:
\begin{equation}\label{eq:tfim}
    \cH=-J\sum_{i=1}^{n-1}Z_iZ_{i+1}-h_x\sum_{j=1}^{n}X_j-h_z\sum_{j=1}^nZ_j
\end{equation}
We benchmark the model on $n=10$ spins. We also select $J=1$ and magnetic field strengths $h_x=1.4$ and $h_z=0.6$, at which point this model is non-integrable.

\subsection{Circuit Implementation and Simulation Methods}\label{sec:sim-methods}
Any implementation of TITE first requires the choice of a gadget to realize $e^{-i\cH t}$ and $e^{-\beta \cH}$. In both cases, we first Trotterize each operator using a $2$nd order Trotter-Suzuki decomposition, with $r_{\rm RTE}=100$ and $r_{\rm ITE}=15$ Trotter steps respectively; the regime in which this decomposition error dominates is assessed in Figure~\ref{fig:sim-trotter-v-ideal}. It then suffices to implement $e^{-iP\Delta t }$ and $e^{-\Delta \beta  P}$, where $P$ is a Pauli term from Eq.~\eqref{eq:tfim}. By appropriately conjugating these operations, we can take either $P=Z_iZ_{i+1}$ or $P=Z_i$. For real-time evolution, $e^{-iP\Delta t}$ can be implemented with a CNOT ladder to compute parity, followed by an $R_Z(\phi_{\rm RTE})$ gate where $\phi_{\rm RTE}=2\Delta t$, and another CNOT ladder to uncompute the parity~\citep{whitfield2011simulation}. Ref.~\cite{leadbeater2024non} presents an almost identical gadget for imaginary-time evolution $e^{-\Delta \beta P}$ by replacing the $R_Z(\phi_{\rm RTE})$ gate with an open-controlled $R_X(\phi_{\rm ITE})$ gate applied to an ancilla, which is measured and post-selected to be $0$, with $\phi_{\rm ITE}=2\arccos( e^{-2\Delta \beta})$. This post-selection is responsible for both the non-unitary nature of $e^{-\Delta \beta P}$ on the system qubits and its associated cost, which is embodied in a compounding post-selection probability as the gadget is repeated throughout the circuit. In any case, for any time $t$ sampled from $\cD$, these gadgets, illustrated in Figure~\ref{fig:rte-ite-gadgets}, allow us to construct a circuit for $e^{-i\cH t}e^{-\beta \cH}$. For the mixed-field Ising model, these circuits solely involve the gates $\{H, X, R_X, R_Z,{\rm CNOT}\}$.

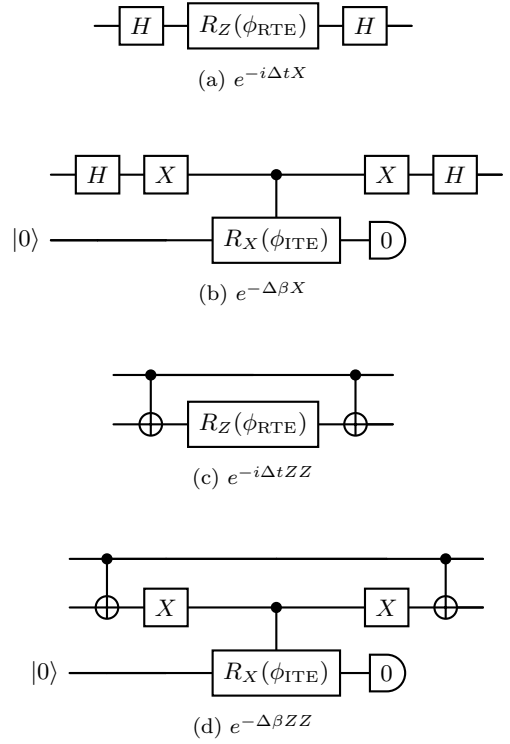
\begin{figure}[t]
\centering
\small


\subfloat[$e^{-i\Delta t X}$]{%
\begin{minipage}{\columnwidth}
\centering
\[
\begin{quantikz}[row sep=0.28cm, column sep=0.33cm]
& \gate{H} & \gate{R_Z(\phi_{\rm RTE})} & \gate{H} & \qw
\end{quantikz}
\]
\end{minipage}}




\subfloat[$e^{-\Delta\beta X}$]{%
\begin{minipage}{\columnwidth}
\centering
\[
\begin{quantikz}[row sep=0.32cm, column sep=0.33cm]
& \gate{H} & \gate{X} & \ctrl{1} & \gate{X} & \gate{H} & \qw \\
\lstick{$\ket{0}$} &
\qw & \qw & \gate{R_X(\phi_{\rm ITE})} & \meterD{0}
\end{quantikz}
\]
\end{minipage}}

\subfloat[$e^{-i\Delta t ZZ}$]{%
\begin{minipage}{\columnwidth}
\centering
\[
\begin{quantikz}[row sep=0.28cm, column sep=0.33cm]
& \ctrl{1} & \qw & \ctrl{1} & \qw \\
& \targ{} & \gate{R_Z(\phi_{\rm RTE})} & \targ{} & \qw
\end{quantikz}
\]
\end{minipage}}


\subfloat[$e^{-\Delta\beta ZZ}$]{%
\begin{minipage}{\columnwidth}
\centering
\[
\begin{quantikz}[row sep=0.32cm, column sep=0.33cm]
& \ctrl{1} & \qw & \qw & \qw & \ctrl{1} & \qw \\
& \targ{} & \gate{X} & \ctrl{1} & \gate{X} & \targ{} & \qw \\
\lstick{$\ket{0}$} &
\qw & \qw & \gate{R_X(\phi_{\rm ITE})} & \meterD{0}
\end{quantikz}
\]
\end{minipage}}

\caption{Circuit gadgets for implementing real-time evolution (RTE) and imaginary-time evolution (ITE) of a single Pauli term, used to Trotterize $e^{-i\cH t}$ and $e^{-\beta\cH}$, respectively, in our simulations. (a) RTE gadget for the single-qubit term $e^{-i\Delta t X}$, implemented via an $R_Z(\phi_{\rm RTE})$ rotation in the $X$ basis. (b) ITE gadget for $e^{-\Delta\beta X}$, implemented by applying an open-controlled $R_X(\phi_{\rm ITE})$ rotation to an ancilla, which is measured and post-selected onto $\ket{0}$. (c)-(d) Analogous RTE and ITE gadgets for the two-qubit term $ZZ$, where a CNOT ladder computes and uncomputes the parity of the two qubits before the same rotation is applied. Throughout, $\phi_{\rm RTE}=2\Delta t$ and $\phi_{\rm ITE}=2\arccos(e^{-2\Delta \beta})$. The post-selected ancilla measurement in (b) and (d) is the source of ITE's non-unitarity.}
\label{fig:rte-ite-gadgets}
\end{figure}

We benchmark these circuits in two complementary regimes. In the \textit{noisy regime}, we optimize the circuit with standard transpiler passes and then wrap every single-qubit and two-qubit gate with a small amount of depolarizing noise:
$$\mathcal{E}(\rho_S)=(1-\gamma)\rho_S +\frac{\gamma}{2^{|S|}} I$$
where $\gamma$ is the noise strength and $\rho_S$ is the reduced density matrix on the subset of qubits $S$ that are involved in the gate. Executing the resulting noisy circuits with a density-matrix simulator captures TITE as it would run on hardware, including gate noise and the finite sampling of $t\sim\cD$. In the \textit{noiseless regime} ($\gamma=0$), we instead forgo circuits entirely: we represent $e^{-\Delta\beta\cH}$ and $e^{-i\cH\Delta t}$ (or their Trotterized approximations) as sparse $2^n\times 2^n$ matrices and compute the output density matrix $\rho=\E_{\cD}[\ketbra{\psi(\beta,t)}]$ by numerically integrating with respect to the density function of $\cD$. This eliminates gate noise and finite-sample effects, allowing us to independently switch on and off each systematic error, namely the Trotterization of either operator and the choice of $\cD$. For clarity, we state the distribution $\cD$ used in each figure.

For reproducibility, we report the computational stack: circuits are constructed and transpiled with Qiskit, noisy executions use the Qiskit Aer density-matrix simulator, and noiseless linear algebra uses NumPy. Simulations were performed on the MIT Engaging cluster, where each data point at $n=10$ requires $\approx 1$--$3$ hours of wall-clock time on $4$ cores with $8\;{\rm GB}$ of memory.

\subsection{Results}\label{sec:sim-results}
\begin{figure*}[t]
    \centering

    \subfloat[Trace distance error and energy vs.\ $\beta$, with the latter compared to the exact ground state energy.\label{fig:sim-noisy-trd-E}]{%
        \includegraphics[width=0.32\textwidth]{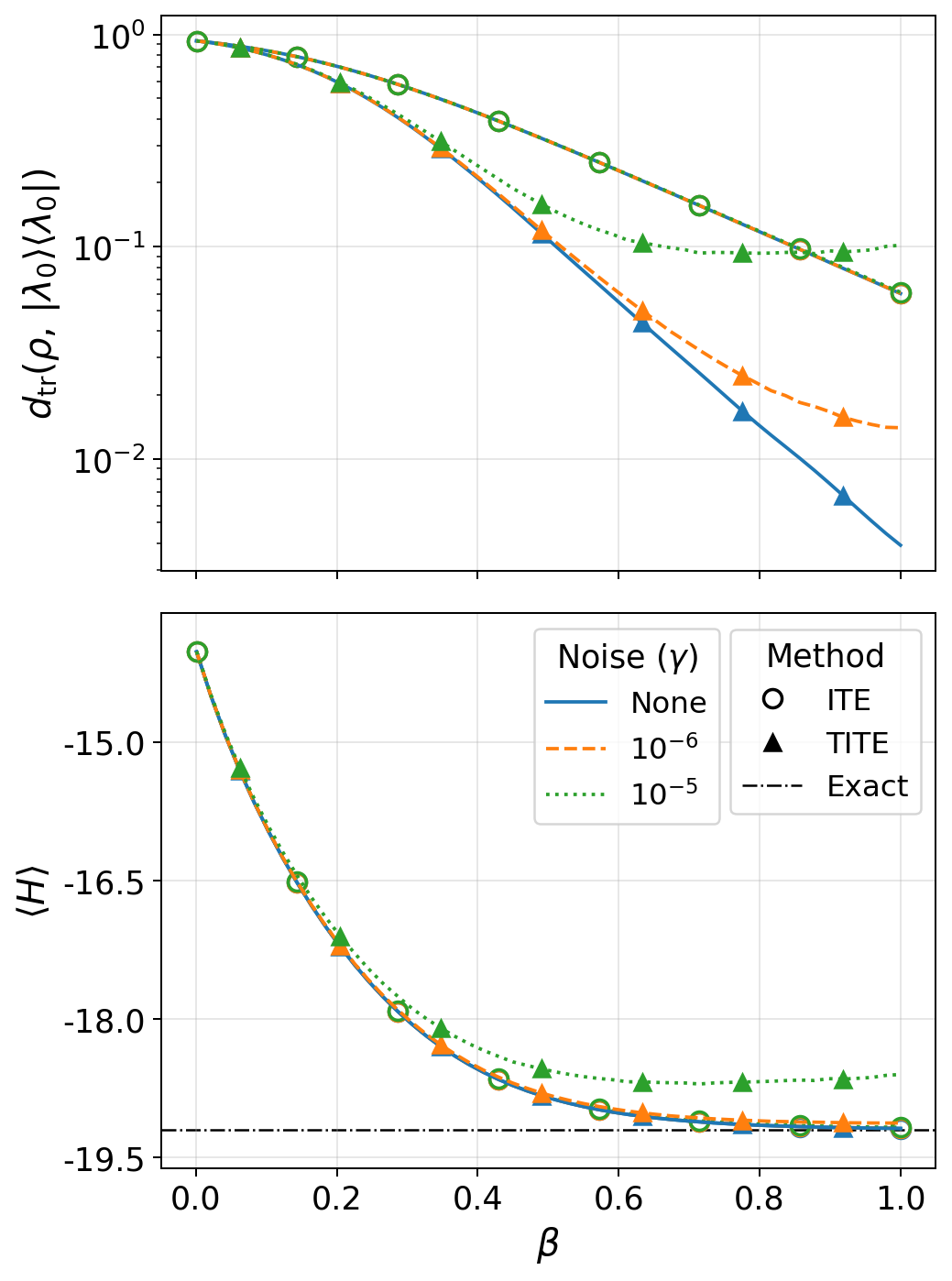}}
    \hfill
    \subfloat[Spin covariance $C_{56}^{xx}$ and $C_{56}^{zz}$ vs.\ $\beta$, compared to their exact values.\label{fig:sim-noisy-scov}]{%
        \includegraphics[width=0.32\textwidth]{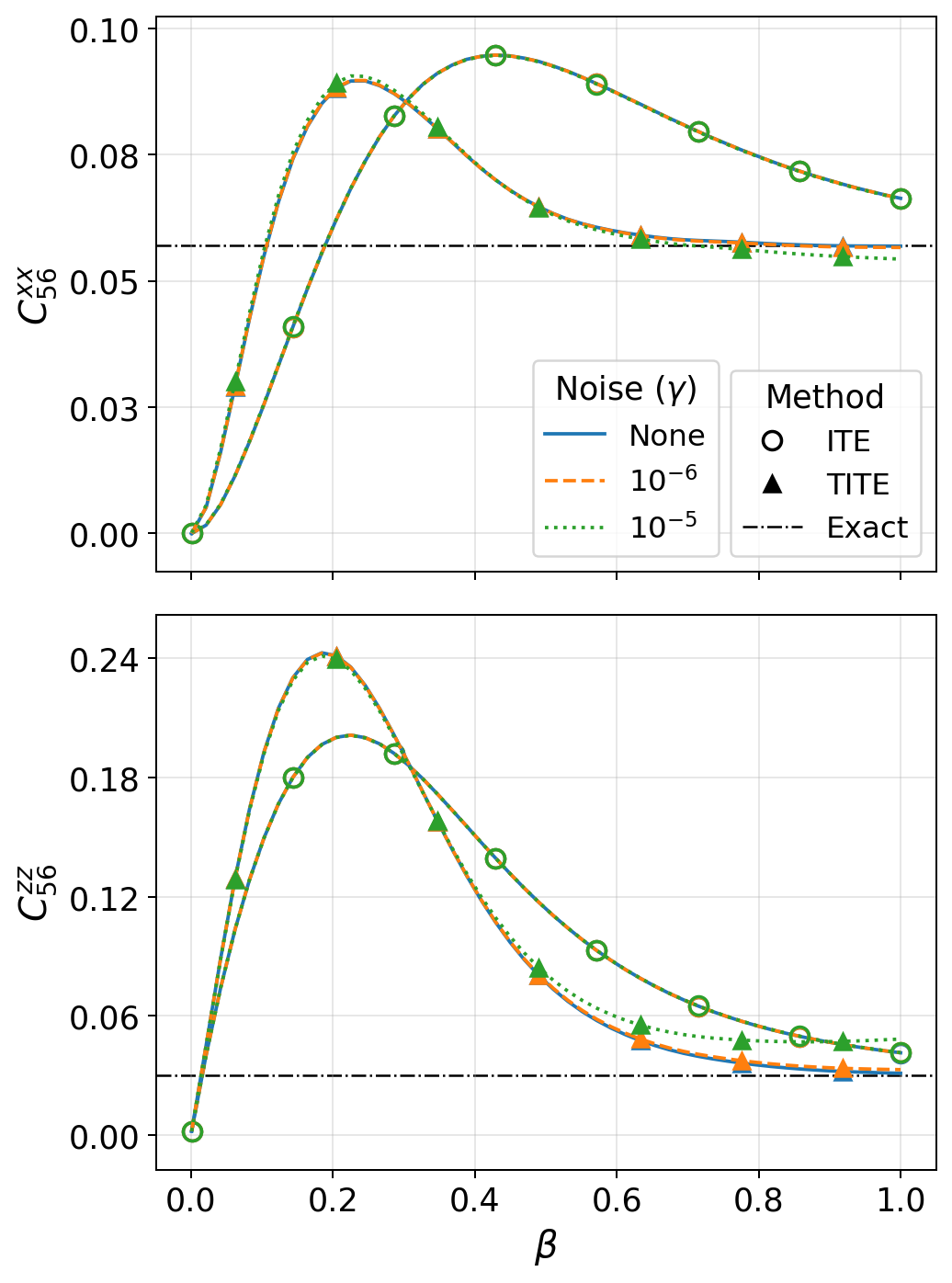}}
    \hfill
    \subfloat[Transverse magnetization, $m_x$, and longitudinal magnetization, $m_z$, vs.\ $\beta$, compared to their exact values.\label{fig:sim-noisy-mag}]{%
        \includegraphics[width=0.32\textwidth]{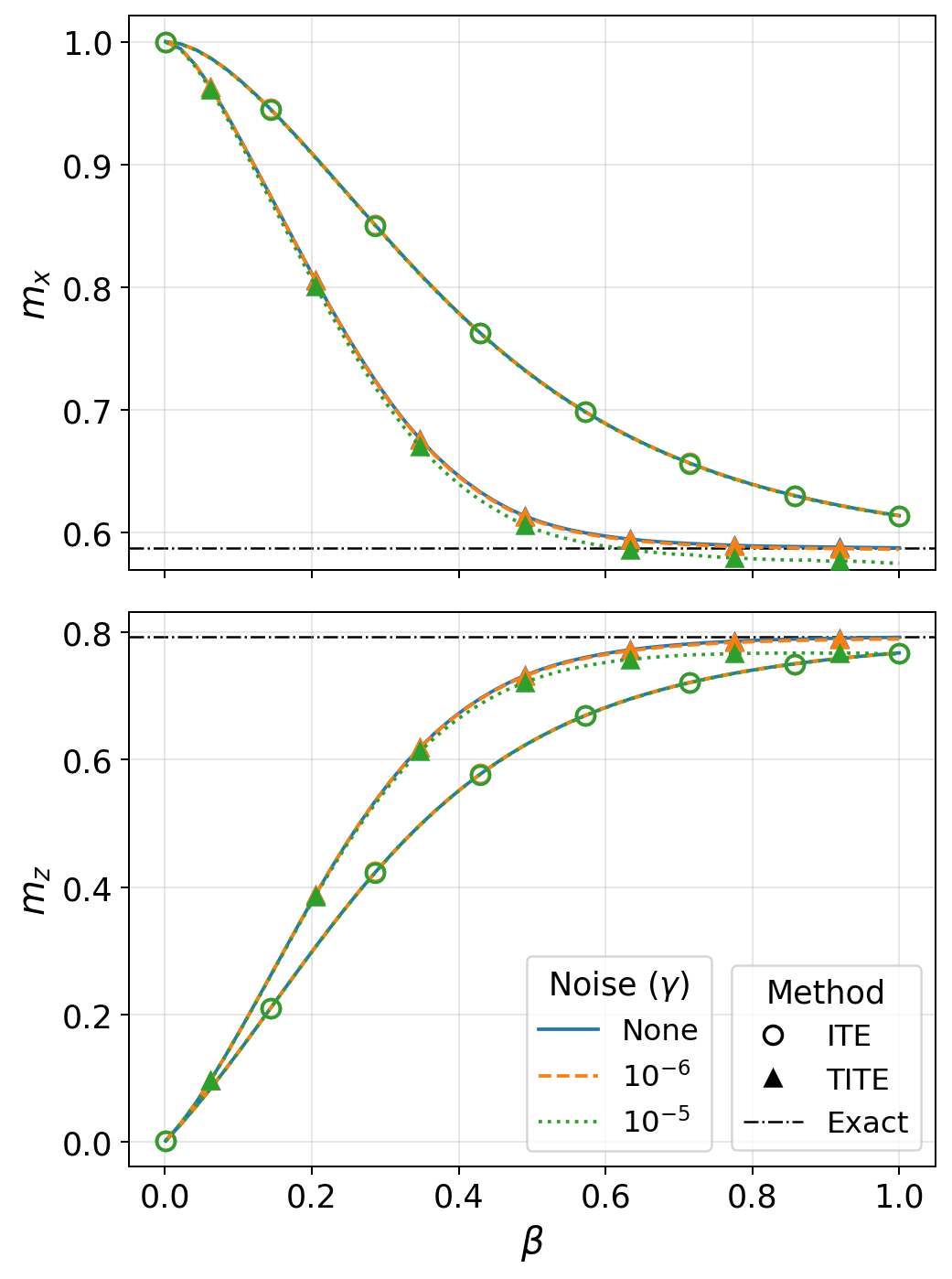}}

    \caption{Noisy circuit simulation ($n=10$), comparing ITE to TITE ($\cD=C_\beta$) across noise strengths $\gamma\in\{0,10^{-6},10^{-5}\}$. }
    \label{fig:sim-noisy}
\end{figure*}

From Figure~\ref{fig:sim-noisy-trd-E}, we see that TITE indeed quadratically suppresses the ground state preparation error, as evidenced by $\log(\trd(\rho,\ketbra{\lambda_0}))$ decreasing twice as fast as that of ITE as $\beta$ increases. Analogously, at any fixed target error $\epsilon$, we confirm that TITE requires roughly half the amount of $\beta$ that imaginary-time evolution does. Importantly, as the noise strength $\gamma$ increases from the noiseless setting to $10^{-6}$ and $10^{-5}$, the advantage of the twirled method worsens yet is still considerable, though the baseline ITE is more noise-robust. This is because the random real-time evolution increases the circuit gate count and depth, which increases the probability that errors propagate and perturb the state.

While standard imaginary-time evolution already achieves $\cO(\epsilon^2)$ error in the energy as shown in Figure~\ref{fig:sim-noisy-trd-E}, the advantage of TITE is that it quadratically suppresses errors in other observables. To illustrate this, Figure~\ref{fig:sim-noisy-scov} plots the spin covariance, $C_{ij}^{rs}\equiv \langle \sigma^r_i \sigma^s_j\rangle -\langle \sigma^r_i\rangle \langle \sigma^s_j\rangle $, between the fifth and sixth sites, which are located in the middle of the $n=10$ spin chain. We see that for both the transverse and longitudinal spin covariance, denoted $C_{56}^{xx}$ and $C_{56}^{zz}$ respectively, TITE converges both closer and sooner to the true observable value (depicted by a dotted black line) than imaginary-time evolution. Figure~\ref{fig:sim-noisy-mag} presents analogous evidence for the transverse and longitudinal magnetization, denoted $m_x$ and $m_z$ respectively, where $m_s=\frac 1 n\sum_{i=1}^{n}\langle \sigma^s_i\rangle $. In fact, it seems that both the spin covariance and magnetization observables are quite robust to gate noise as increasing $\gamma$ does not sizably worsen the estimate. Indeed, this contrast between the energy and other observables is the structure versus randomness dichotomy described in Sec.~\ref{sec:stabilizers}: the energy is a structured observable, diagonal in the energy eigenbasis, and therefore oblivious to the coherent part of the error, whereas the covariance and magnetization experience first-order errors which randomization renders incoherent. Further, Figures~\ref{fig:sim-noisy-scov} and \ref{fig:sim-noisy-mag} show that TITE gives a greater advantage over ITE for $X$ observables than $Z$ observables. This is because the latter have a smaller commutator with our choice of Hamiltonian, and are thus less affected by coherent error.

Having established quadratic error suppression in the noisy regime, where our simulation choices ensured that Trotterization error and the statistical error from a finite number of sampled times $t\sim\cD$ are not dominant, we now turn to the noiseless regime to isolate these two systematic effects.
Figure~\ref{fig:sim-trotter-v-ideal} isolates the effect of Trotterizing $e^{-\beta \cH}$ from the effect of Trotterizing $e^{-i\cH t}$ by using either or both of the operators in their exact matrix form. We see that when only imaginary-time evolution is Trotterized, beyond some critical inverse temperature $\beta_{\rm ITE}^c$, the Trotterization error dominates the $\exp(-\beta \Delta)$ ground state preparation error arising from imaginary-time evolution itself. Thus, we see that both imaginary-time evolution and TITE exhibit the same critical $\beta_{\rm ITE}^c$ beyond which the trace distance plateaus. On the other hand, when only real-time evolution is Trotterized, TITE exhibits quadratic error suppression up to some critical $\beta_{\rm RTE}^c$ after which the Trotterization error again dominates the suppressed ground state preparation error $\sim\exp(-2\beta \Delta)$, causing the trace distance to plateau. Just as one must traditionally ensure that imaginary-time evolution Trotterization error does not dominate, TITE requires that the real-time evolution Trotterization error not dominate either. 
\begin{figure}
    \centering
    \includegraphics[width=\linewidth]{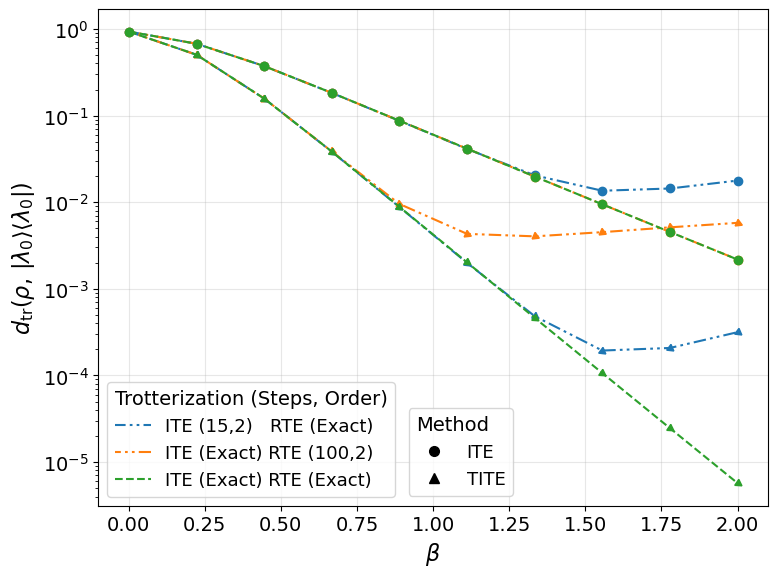}
    \caption{Noiseless comparison isolating Trotterization error in ITE and RTE ($n=10$, $\cD=\mathcal N_{\beta,\Delta}$). A plot of trace distance vs.\ $\beta$ when either $e^{-\beta\cH}$, $e^{-i\cH t}$, or neither is Trotterized shows that Trotter error begins to dominate preparation error beyond the critical points $\beta_{\rm ITE}^c$ and $\beta_{\rm RTE}^c$. As shown in the caption, the numbers in parentheses are (number of Trotter steps, order of Trotterization).}
    \label{fig:sim-trotter-v-ideal}
\end{figure}

\begin{figure}
    \centering
    \includegraphics[width=1\linewidth]{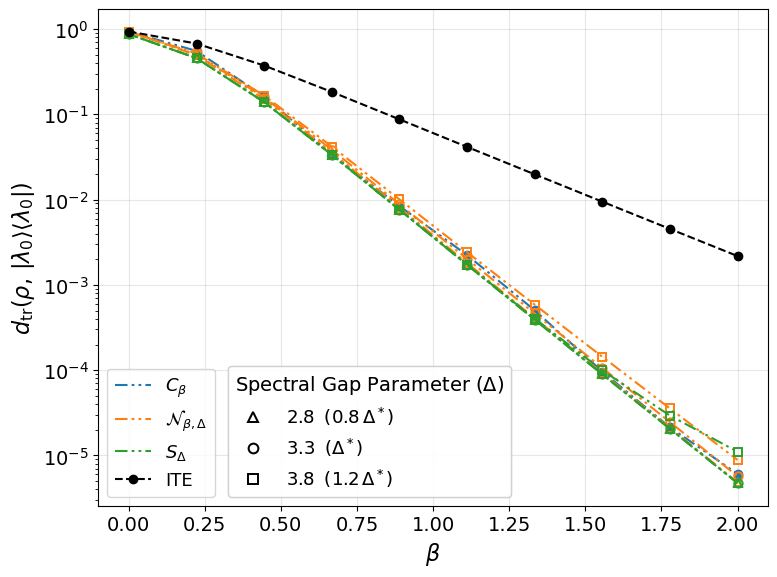}
    \caption{Noiseless comparison of the three twirling distributions $S_\Delta$, $C_\beta$, and $\mathcal N_{\beta,\Delta}$ ($n=10$). A plot of trace distance vs.\ $\beta$ shows that $S_\Delta$ achieves the best error suppression, where $\Delta$ is the spectral gap parameter and $\Delta^*$ is the true spectral gap.}
    \label{fig:sim-distrib}
\end{figure}

In Figure~\ref{fig:sim-distrib}, we compare the performance of all three distribution choices presented in Section~\ref{sec:distrib-choice}. The dominant feature is the robustness of the advantage: all three distributions double the decay rate of the trace distance, from the $e^{-\beta\Delta}$ of ITE to the $e^{-2\beta\Delta}$ guaranteed by Theorem~\ref{thm:ground-qite-error-suppression}, so that by $\beta=2$ the TITE preparation error lies nearly three orders of magnitude below that of ITE, an accuracy that pure imaginary-time evolution would attain only at $\beta \approx 4$, at quadratically greater sample cost. Among the distributions, $S_\Delta$ performs best while also yielding optimal cost, as it is constructed so that $|\varphi_\cD(\omega)|=0$ for $\omega \geq \Delta$, exactly cancelling the phase factors induced by randomized real-time evolution on the excited states. Yet the margin is small: even $C_\beta$, which requires no knowledge of the spectral gap whatsoever, and $\mathcal N_{\beta,\Delta}$ with the gap parameter $\Delta$ mis-estimated by $\pm 20\%$ relative to the true spectral gap, remain within a factor of $\approx 2$ of optimal performance. The quadratic advantage is thus a property of the randomization itself, not of a finely tuned distribution. 
Moreover, the TITE curves do not merely obey the $e^{-2\beta\Delta}$ bound of Theorem~\ref{thm:ground-qite-error-suppression}; they saturate the matching error lower bound presented in Appendix~\ref{x:rc-quadratic-optimality}, confirming that randomization removes the coherent content of the error entirely.

\begin{figure}
    \centering
    \includegraphics[width=\linewidth]{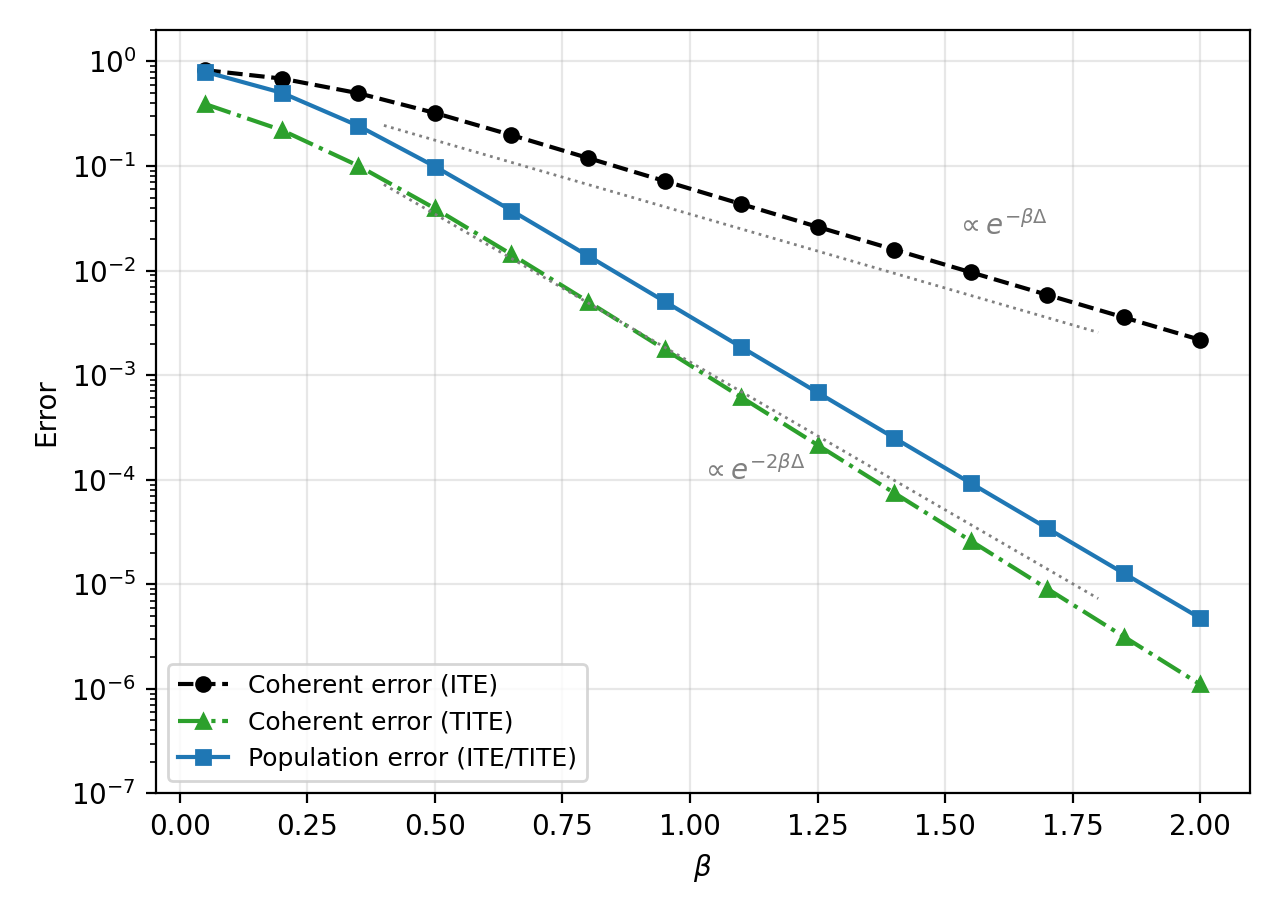}
    \caption{Noiseless simulation of the population error and coherent error ($n=10$, $\cD=S_\Delta$). The population error (blue) is shared by ITE and TITE and decays as $e^{-2\beta \Delta}$, as it is invariant under randomized real-time evolution. The coherent error of ITE (black) decays only as $e^{-\beta\Delta}$, whereas randomization provided by TITE suppresses the coherent error (green) down to the level of the population error.}
    \label{fig:sim-coherence}
\end{figure}

Figure~\ref{fig:sim-coherence} makes this structure versus randomness decomposition directly visible by separately depicting the diagonal \textit{population error} and the off-diagonal coherent error, whose sum upper bounds the ground state trace distance by the triangle inequality:
\begin{align*}
    &\trd(\rho,\ketbra{\lambda_0})\leq \underbrace{\frac 1 2 \sum_{k}|\rho_{kk}-\delta_{k0}|}_{\text{ population error}}+\underbrace{\frac 1 2 \|\rho-{\rm diag}(\rho)\|_1}_{\text{coherent error}},
\end{align*}
where $\rho_{kk}$ and ${\rm diag}(\rho)$ view $\rho$ as a matrix in the energy basis $\{\ket{\lambda_k}\}$. We see from Figure~\ref{fig:sim-coherence} that real-time evolution leaves the populations untouched, leading to the same population error term for ITE and TITE which decays at a rate $\propto e^{-2\beta\Delta}$; it is precisely this structured content of the error that is responsible for the error lower bound of Appendix~\ref{x:rc-quadratic-optimality}.
By contrast, the coherent error of ITE decays only as $e^{-\beta\Delta}$ and dominates its total error, while TITE suppresses the coherent error down to the level of the population error.

As a remark on simulation methodology, we find that computing $\rho$ via numerical integration is easiest for $\mathcal N_{\beta,\Delta}$, followed by $C_\beta$, and $S_\Delta$ by a large margin. This is because normal distributions are much better behaved than the artisanal $C_\beta$ and $S_\Delta $ distributions for which higher order moments do not exist, contributing to longer integration convergence times. In noisy circuit simulations, this manifests itself in finite sample effects from estimating $\rho$ as an average of density matrices, each corresponding to a simulated circuit with a drawn random time $t\sim \cD$. It is for this reason that our noisy simulation in Figure~\ref{fig:sim-noisy} uses $C_\beta$ and $\beta \in (0,1/2)$, such that the finite sample effects do not dominate the $\exp(-\beta \Delta)$ scale of the ITE error. We emphasize that these convergence considerations are artifacts of classical simulation, where the mixed state $\rho=\E_{\cD}[\ketbra{\psi(\beta,t)}]$ must be explicitly assembled by averaging over $t$. On quantum hardware, no such averaging step is performed: each shot simply draws a time $t\sim\cD$ and executes the corresponding circuit, and the resulting measurement statistics are automatically those of $\rho$. The choice of distribution therefore affects the hardware cost only through the expected real-time evolution cost $\E_{\cD}[|t|^\alpha]$, and not through any additional sampling overhead.

\section{Discussion and Conclusions}\label{sec:conclusion}

In this work, we developed twirled imaginary-time evolution as a general framework for enhancing imaginary-time evolution with randomized real-time dynamics. Concretely, we proved that appending real-time evolution to ITE for a random duration $t\sim\cD$ (for an appropriately chosen distribution $\cD$) quadratically suppresses the trace distance to the ground state, from $\cO(\epsilon)$ to $\cO(\epsilon^2)$. Equivalently, the imaginary time required to reach a fixed accuracy is halved, $\beta \to \beta/2$, so the $e^{\cO(\beta)}$ sample complexity/classical cost of an ITE implementation is quadratically reduced.
Our numerical demonstrations on a non-integrable spin chain confirm this picture in detail: the error decay rate doubles, the advantage persists under gate noise and Trotterization, and the algorithm is robust to the choice of distribution, requiring neither fine tuning nor precise knowledge of the spectral gap.

Several conditions limit the scope of these results. The method presupposes a gapped, non-degenerate Hamiltonian and an initial state with non-negligible overlap with the ground state, with performance that degrades as either assumption weakens. For instance, a degenerate ground state subspace would require randomization conditioned on the gap of the target subspace. Moreover, our analysis addressed algorithmic error only, and while our simulations show the advantage surviving moderate depolarizing noise, the deeper circuits incurred by real-time evolution can amplify error propagation. A complete noise analysis, including whether temporal twirling inherits a Pauli-twirling-like noise resilience~\cite{Wallman_2016, vandenBerg_2022}, remains for future work. 

The central idea of using randomized real-time evolution to suppress errors applies well beyond ITE, in both quantum and classical settings. Randomized real-time evolution could enhance ground state estimation in variational quantum algorithms without increasing the number of variational parameters, and may similarly benefit other ground state preparation algorithms~\cite{Lin2020Near, Dong_2022, chen2023quantum, Zhan_2026}. On the classical side, ongoing work is exploring the utility of randomized time evolution in variational Monte Carlo, where it enables more accurate estimates of ground state observables~\cite{Martyn_Luo_2025}. The generalization to stabilizer states discussed in Sec.~\ref{sec:stabilizers} further extends our method to the broader task of projecting into a subspace stabilized by an operator, with potential applications to virtual distillation~\cite{Huggins_2021}, state preparation protocols, and algorithmic cooling~\cite{Laflemme_2022}.

Another particularly intriguing direction is to make the algorithm fully unitary by purifying its randomness. Instead of classically sampling $t\sim\cD$, one may prepare a clock register~\cite{kitaev2002classical} in the coherent state $\ket*{\sqrt{\cD}} \equiv \int dt\, \sqrt{f_{\cD}(t)}\, \ket{t}$ and apply the controlled evolution $\int dt\, \ketbra{t}\otimes e^{-i\cH t}$, where $f_{\cD}$ is the probability density function of $\cD$. Tracing out the clock register reproduces our channel exactly, but retaining it makes explicit that the entropy generated by the algorithm (upon tracing out the clock register) equals the entanglement entropy between the system and clock. Hence, our randomization does not destroy the coherent error, but rather transfers it into system-clock correlations, and the Landauer cost of resetting the clock register quantifies the thermodynamic price of the randomness. 
Better understanding this equivalence may help clarify exactly how much of ground state preparation can be delegated to randomness and at what thermodynamic cost.

Taken together, the results presented here add to the growing body of evidence that randomization, whether sampled classically or purified into a quantum register, is a powerful resource for enhancing quantum information processing.

\begin{acknowledgements}
This research was supported by PNNL's Quantum Algorithms and Architecture for Domain Science (QuAADS) Laboratory Directed Research and Development (LDRD) Initiative. This material is based upon work supported by the U.S. Department of Energy, Office of Science, National Quantum Information Science Research Centers, Quantum Science Center (QSC). The Pacific Northwest National Laboratory is operated by Battelle for the U.S. Department of Energy under Contract DE-AC05-76RL01830.
\end{acknowledgements}

\section*{Data Availability}
The data and source code that support the findings of this article are
openly available~\cite{arulandu_tite_github}.

\appendix
\section{Optimality of Quadratic Suppression}\label{x:rc-quadratic-optimality}
In this section, we show that any randomized strategy similar to twirled imaginary-time evolution cannot suppress ground state preparation errors beyond $\cO(\epsilon^2)$. 
\begin{lemma}\label{lemma:quadratic-optimality-rite}
    Consider any ground state approximation of the form: 
    $$\ket*{\psi}=\sqrt{1-\delta^2}\ket{\lambda_0}+\delta \ket{\lambda_\perp}$$
    where $\braket{\lambda_0}{\lambda_\perp}=0$. Further, let $\{U_t\}_{t\in T}$ be a collection of unitaries indexed by $t\in T$ for a set $T$ that stabilize the ground state up to phase:
    $$U_t\ket{\lambda_0}=e^{i\theta_t}\ket{\lambda_0}.$$
    Let $\cD$ be a distribution over $T$, and consider the mixed state corresponding to the stochastic process of drawing $t\sim \cD$ and applying $U_t$ to $\ket{\psi}$:
    $$\rho\equiv\E_{ \cD}[U_t\ketbra{\psi}U_t^\dagger]$$
    Then, $\trd(\rho,\ketbra{\lambda_0})\geq \delta^2$. Applied to the imaginary-time-evolved state $\ket{\psi(\beta)}$ from Theorem~\ref{thm:ground-qite-error-suppression}, this shows that any twirling procedure using unitaries that stabilize the ground state up to phase cannot achieve a super-quadratic error suppression. 
\end{lemma}
\begin{proof}
    Consider the projective measurement onto the ground state: $\{\ketbra{\lambda_0},I-\ketbra{\lambda_0}\}$. Then, 
    \begin{align*}
        \tr(\ketbra{\lambda_0}\rho)&=\E_{\cD}[\bra{\lambda_0}U_t\ketbra{\psi}U_t^\dagger \ket{\lambda_0}]\\
        &=\E_{\cD}[e^{i\theta_t}\bra{\lambda_0}\cdot \ketbra{\psi}\cdot e^{-i\theta_t}\ket{\lambda_0}]\\
        &=\E_{\cD}[\lvert\braket{\lambda_0}{\psi}\rvert^2]\\
        &=1-\delta^2 
    \end{align*}
    where we recall $\ket{\psi}=\sqrt{1-\delta^2}\ket{\lambda_0}+\delta\ket{\lambda_\perp}$. Then, applying this projective measurement to $\rho$ and $\ketbra{\lambda_0}$ gives classical outcome distributions ${\rm Bern}(\delta^2)$ and ${\rm Bern}(0)$ respectively, where Bern denotes a Bernoulli distribution. Then, 
    \begin{align*}
       \trd(\rho,\ketbra{\lambda_0})\geq d_{\rm TV}({\rm Bern}(\delta^2),{\rm Bern}(0))=\delta^2,
    \end{align*}
    because the trace distance between two states upper bounds the total variation distance between the outcome distributions of any POVM applied to the two states. 
\end{proof}

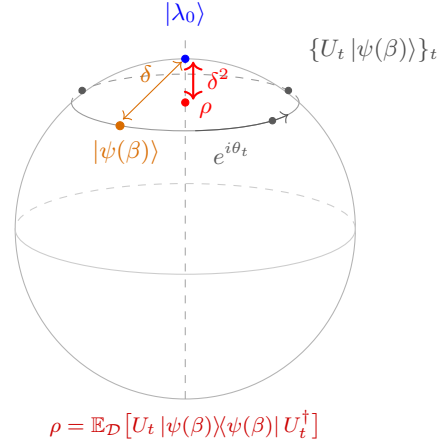
\begin{figure}[t]
\centering
\usetikzlibrary{calc}
\begin{tikzpicture}[scale=0.75, every node/.style={font=\small}]


\draw[gray!60] (0,0) circle (3);
\draw[gray!40]  (-3,0) arc (180:360:3 and 0.8);   
\draw[gray!40, dashed] (3,0) arc (0:180:3 and 0.8); 

\draw[gray!70, dashed] (0,-3) -- (0,3.45);

\draw[gray!80]          (-2.007,2.23) arc (180:360:2.007 and 0.5); 
\draw[gray!80, dashed]  (2.007,2.23)  arc (0:180:2.007 and 0.5);   

\foreach \u in {320, 25, 155} {
  \filldraw[gray!70!black] ($(0,2.23)+(\u:2.007 and 0.5)$) circle (1.6pt);
}

\draw[->, gray!70!black] ($(0,2.23)+(275:2.007 and 0.5)$)
    arc (275:335:2.007 and 0.5);
\node[gray!70!black, anchor=north] at ($(0,2.23)+(300:1.65 and 0.62)$)
    {$e^{i\theta_t}$};

\filldraw[orange!85!black] (-1.151,1.82) circle (2pt);
\node[anchor=east, text=orange!85!black] at (-0.28,1.36) {$\ket{\psi(\beta)}$};

\draw[<->, orange!85!black, thin] (-1.10,1.93) -- (-0.06,2.95);
\node[orange!85!black, anchor=south east] at (-0.42,2.42) {$\delta$};

\filldraw[red] (0,2.23) circle (1.8pt);
\node[anchor=north west, text=red] at (0.08,2.40) {$\rho$};

\draw[<->, red, thick] (0.14,2.28) -- (0.14,2.95);
\node[red, anchor=west] at (0.22,2.64) {$\delta^2$};

\filldraw[blue] (0,3) circle (1.8pt);
\node[anchor=south, text=blue] at (0,3.42) {$\ket{\lambda_0}$};

\node[gray!60!black, anchor=west] at (2.0,3.15) {$\{U_t\ket{\psi(\beta)}\}_{t}$};

\node[red!80!black, font=\footnotesize] at (0,-3.45)
    {$\rho=\E_{\cD}\big[U_t \ketbra{\psi(\beta)}{\psi(\beta)} U_t^\dagger\big]$};

\end{tikzpicture}
\caption{Optimality of quadratic error suppression (Lemma~\ref{lemma:quadratic-optimality-rite}). In the effective two-level space spanned by the ground state $\ket{\lambda_0}$ (north pole) and the excited component $\ket{\lambda_\perp}$, any set of unitaries $U_t$ that stabilize the ground state up to phase preserves the condition $\bra{\lambda_0}\rho\ket{\lambda_0} = 1-\delta^2$, so the states $\{U_t \ket{\psi(\beta)}\}_t$ form an orbit about $|\lambda_0\rangle$ at fixed latitude (polar angle exaggerated for visibility). Mixing over these states erases the coherent error $\delta$ (orange chord) and produces a mixed state $\rho$ that lies toward the center of the sphere. However, the state can never move further upward as the condition $\bra{\lambda_0}\rho\ket{\lambda_0} = 1-\delta^2$ remains true. This sets the floor $\trd(\rho, \ketbra{\lambda_0}{\lambda_0}) \geq \delta^2$.
}
\label{fig:optimality}
\end{figure}

Figure~\ref{fig:optimality} geometrically illustrates this limitation for the imaginary-time-evolved state $\ket{\psi(\beta)}$, showing that while mixing maps the states $\{U_t\ket{\psi(\beta)}\}$ closer to the true ground state, it can never suppress the error beyond $\delta^2$. Lemma~\ref{lemma:quadratic-optimality-rite} thus implies that the only way, if any, to achieve super-quadratic error suppression is with operations that do not stabilize the ground state. This means that we must reach beyond real-time evolution to operations that act on the excited states themselves. One candidate approach is to modify the purification of TITE discussed in Sec.~\ref{sec:conclusion}; namely, instead of discarding the clock after the controlled evolution $\int dt\ketbra{t}\otimes e^{-i\cH t}$, we can measure the clock in a Fourier basis and post-select the outcome, effectively yielding a spectral filter. This notably does not stabilize the ground state up to phase, circumventing the lower bound in Lemma~\ref{lemma:quadratic-optimality-rite}. Of course, post-selection on the clock incurs sample complexity cost, but if the resulting error suppression is strongly super-quadratic, this may be a more efficient use of samples for ground state preparation than an equivalently longer amount of imaginary-time evolution followed by temporal twirling.

\section{Distributions}\label{x:distrib}
Referring to Section~\ref{sec:distrib-choice}, this appendix establishes the properties of the three randomization distributions presented. When a lower bound $\Delta$ on the spectral gap is known, Lemma~\ref{lemma:sinc-distrib} constructs $S_\Delta$, whose characteristic function vanishes identically beyond $\Delta$ at cost $\E[|t|^\alpha]=\Theta(\Delta^{-\alpha})$, and Lemma~\ref{lemma:distrib-opt} proves this cost optimal. When $\Delta$ is unknown, Lemma~\ref{lemma:gcauchy} constructs $C_\beta$, whose characteristic function decays as $e^{-\beta|\omega|}$ at cost $\Theta(\beta^\alpha)$, and Corollary~\ref{cor:gcauchy-opt} proves this cost is optimal in turn. Finally, Lemma~\ref{lemma:normal-distrib} shows that an ordinary Gaussian $\mathcal{N}_{\beta,\Delta}$ also satisfies Eq.~\eqref{eq:D-phi-decay}, at the suboptimal cost $\Theta((\beta/\Delta)^{\alpha/2})$. The three characteristic functions are compared schematically in Figure~\ref{fig:distrib-schematic}, their empirical performance in Figure~\ref{fig:sim-distrib}, and their exact densities and characteristic functions are plotted at representative parameters in Figure~\ref{fig:distrib-exact}, where each of the results below appears as a visible feature.

Throughout, we adopt the convention $\varphi_{\cD}(\omega) \equiv \E_{t\sim\cD}[e^{-i\omega t}] = \int f_{\cD}(t)\, e^{-i\omega t}\, dt = \hat f_{\cD}(\omega)$ for the characteristic function of a distribution $\cD$ with density $f_{\cD}$, and write ${\mathbb I}(\cdot)$ for the indicator function.

\begin{figure*}[t]
    \centering
    \includegraphics[width=0.92\textwidth]{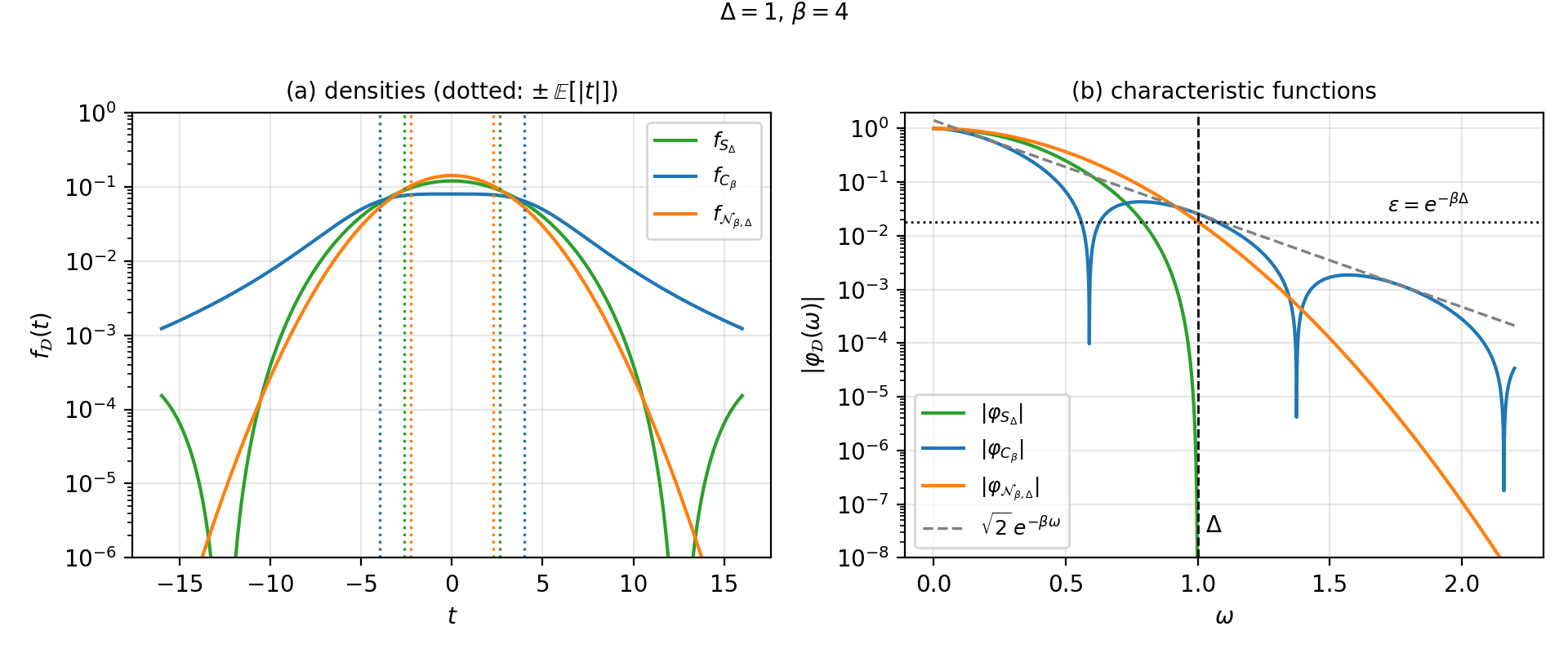}
    \caption{The three twirling distributions of this appendix, plotted exactly with $\Delta = 1$ and $\beta = 4$, for which $\epsilon = e^{-\beta\Delta} \approx 0.018$. (a) The densities $f_{\cD}(t)$, with dotted lines marking the exact mean absolute times: $\E[|t|] = 12\ln 2/(\pi\Delta) \approx 2.65$ for $S_\Delta$ (Lemma~\ref{lemma:sinc-distrib}), $ 4$ for $C_\beta$ (Lemma~\ref{lemma:gcauchy}), and $4/\sqrt{\pi} \approx 2.26$ for $\mathcal N_{\beta,\Delta}$ (Lemma~\ref{lemma:normal-distrib}). This plot evidences the fact that the spread of $S_\Delta$ is set by $1/\Delta$ alone, independent of $\beta$. (b) The characteristic functions $|\varphi_{\cD}(\omega)|$, each exhibiting its respective lemma: $\varphi_{S_\Delta}$ vanishes identically at $\omega = \Delta$ (Lemma~\ref{lemma:sinc-distrib}); $\varphi_{C_\beta}$ oscillates beneath its exponential envelope (dashed gray) with cusps at the cosine zeros (Lemma~\ref{lemma:gcauchy}); and the Gaussian crosses the decay threshold at exactly $\omega = \Delta$ (Lemma~\ref{lemma:normal-distrib}). As required by Eq.~\eqref{eq:D-phi-decay}, all three fall to $\cO(\epsilon)$ (dotted guide) by $\omega = \Delta$ (dashed guide), up to the constant $\sqrt 2$ in the case of $C_\beta$.}
    \label{fig:distrib-exact}
\end{figure*}
\begin{lemma}\label{lemma:sinc-distrib}
Let $S_\Delta$ be the distribution defined by the density $f_{S_\Delta}(t)\equiv \frac{3\Delta}{8\pi}\,{\rm sinc}^4(\Delta t/4)$. Then, $\varphi_{S_\Delta}(\omega)=0$ for $|\omega|\geq \Delta$, and its absolute moments are
\begin{align*}
    \E_{S_\Delta}[|t|^\alpha]&=\frac{12(2^\alpha-2)\Gamma(\alpha-3)\sin(\alpha\pi/2)}{\pi \Delta^\alpha }\\
    &\propto\begin{cases}
     \Delta^{-\alpha} & \alpha \in [1,3)\\
    \infty& \alpha \in [3,\infty)
\end{cases}
\end{align*}
where the formula is to be interpreted as a limit at $\alpha \in \{1,2\}$ (where the pole of $\Gamma(\alpha-3)$ cancels the zero of its cofactor), giving the values used in the main text:
\begin{align*}
    \E_{S_\Delta}[|t|]=\frac{12\ln 2}{\pi \Delta}, \qquad \E_{S_\Delta}[t^2]=\frac{12}{\Delta^2}.
\end{align*}
\end{lemma}
\begin{proof}
    We first prove the property, ${\rm supp}(\varphi_{S_\Delta})\subset [-\Delta,\Delta]$, which explains the choice of ${\rm sinc}^4(\cdot)$. Define
    $$g(t)\equiv \frac{\Delta}{4\pi}{\rm sinc}(\Delta t/4),$$
    whose Fourier transform is the rectangle $\hat g (\omega)={\mathbb I}(\omega \in [-\Delta/4,\Delta/4])$. Since $f_{S_\Delta} \propto g^4$, we have
    \begin{align*}
        \varphi_{{S_\Delta}}(\omega)&\propto \hat g (\omega)*\hat g(\omega)*\hat g(\omega)*\hat g(\omega),
    \end{align*}
    where $*$ denotes convolution: the four-fold convolution of rectangles of width $\Delta/2$, which is supported in $[-\Delta,\Delta]$. Carrying out the convolution explicitly, with a suitable normalization so that $\varphi_{S_\Delta}(0)=1$, yields the cubic B-spline:
    \begin{align*}
        \varphi_{S_\Delta}(\omega) = \begin{cases}
            1-\tfrac 3 2 u^2 + \tfrac 3 4 u^3 & u\leq 1\\
            \tfrac 1 4 (2-u)^3 & 1\leq u \leq 2\\
            0 & u \geq 2
        \end{cases}, \quad u \equiv \frac{2|\omega|}{\Delta}.
    \end{align*}
    This closed form is convenient for numerical noiseless simulations of TITE.

    For normalization and the moments, we recall two standard integral identities, which appear as Equations 3.827(7) and 3.823(1) respectively in Ref.~\citep{gradshteyn2014table}:
    \begin{align*}
        \int_0^\infty \frac{\sin^4(at)}{t^4}dt&=\frac{a^3 \pi}{3}\\
        \int_0^\infty t^{{\kappa}-1}\sin^2 (at)dt&=-\frac{\Gamma({\kappa})\cos(\pi {\kappa}/2)}{2^{{\kappa}+1}a^{\kappa}}
    \end{align*}
    where $a>0$ and ${\kappa}\in (-2,0)$.
    Applying the first for $a=\Delta/4$ with the appropriate constant factors, we confirm that $f_{{S_\Delta}}$ is indeed a valid probability density. As for the moments, using $u=\Delta t/4$,
    \begin{align*}
    \E_{S_\Delta }[|t|^\alpha]&=\frac{3\cdot 2^{2\alpha-1}}{\pi \Delta^\alpha}\int_{-\infty}^{\infty}|u|^{\alpha}{\rm sinc}^4 u\; du \\
    &=\frac{3\cdot 2^{2\alpha}}{\pi \Delta^\alpha }\int_0^{\infty }\frac{\sin^4 u}{u^{4-\alpha}}du
    \end{align*}
    Since $2\sin^2(u/2)=1-\cos(u)$, from the second identity for $\kappa=\alpha-3$, we have that:
    \begin{align*}
        \int_0^{\infty}\frac{1-\cos(au)}{u^{4-\alpha}}\;du&=\frac{\pi a^{3-\alpha}}{2\Gamma(4-\alpha)\sin(\pi(3-\alpha)/2)}
    \end{align*}
    Then, applying this following the fact that $\sin^4 u =\frac 1 2(1-\cos(2u))-\frac 1 8 (1-\cos(4u))$, we have that:
    \begin{align*}
        &\E_{S_\Delta }[|t|^\alpha]\\
        &=\frac{3\cdot 2^{2\alpha}}{\pi \Delta^\alpha }\frac{\pi }{2\Gamma(4-\alpha)\sin(\pi(3-\alpha)/2)}\left(\frac{2^{3-\alpha}}{2}-\frac {4^{3-\alpha}}{8}\right)\\
        &=\frac{6(2^\alpha-2)}{\Delta^\alpha\Gamma(4-\alpha)\sin(\pi(3-\alpha)/2)}\\
        &=\frac{6(2^\alpha-2)\Gamma(\alpha-3)}{\pi \Delta^\alpha}\cdot \underbrace{\frac{\sin(\pi(\alpha-3))}{\sin(\pi(3-\alpha)/2)}}_{2\sin(\pi \alpha/2)}
    \end{align*}
    where the final line applies Euler's reflection formula for the Gamma function followed by standard trigonometric identities. Taking $\alpha \to 1$ and $\alpha \to 2$ (where $(2^\alpha - 2)\Gamma(\alpha-3)\sin(\alpha\pi/2) \to \ln 2$ and $\pi$, respectively) produces the stated values of $\E_{S_\Delta}[|t|]$ and $\E_{S_\Delta}[t^2]$; the moment diverges at $\alpha = 3$ (logarithmically) and beyond, as the density has tails $\sim |t|^{-4}$.
\end{proof}
Thus, $\cD=S_{\Delta}$ is indeed a distribution that satisfies the desired characteristic tail condition in Eq.~\eqref{eq:D-phi-decay} with absolute moments $\E_{\cD}[|t|^\alpha]=\Theta(\Delta^{-\alpha})$, which are independent of the target error $\epsilon$. We next show that the corresponding real-time implementation overhead is minimal up to constant factors by providing a lower bound for these absolute moments.

\begin{lemma}\label{lemma:distrib-opt}
Any distribution $\cD$ satisfying Eq.~\eqref{eq:D-phi-decay} has moment scaling $\E_{\cD}[|t|^\alpha]= \Omega(\Delta^{-\alpha})$ for $\alpha \in [1,2]$.
\end{lemma}
\begin{proof}
    Applying Eq.~\eqref{eq:D-phi-decay} for $\omega=\Delta$, 
    \begin{align*}
    1-c\epsilon &\leq 1-|\varphi_{\cD}(\Delta)|\leq |1-\varphi_{\cD}(\Delta)|\leq \E_{\cD}[|1-e^{-i\Delta t}|]\\
    &\leq \Delta \E_{\cD}[|t|]
    \end{align*}
    where we use the fact that $|1-e^{-ix}|=2|\sin(x/2)|\leq |x|$, with $c$ being the constant factor in \eqref{eq:D-phi-decay}. Similarly, $\Re[\varphi_{\cD}(\omega)]\leq |\varphi_{\cD}(\omega)|\leq c\epsilon$, meaning
    \begin{align*}
       1-c\epsilon &\leq 1 -\Re[\varphi_{\cD}(\Delta)] =\E_{\cD}[1-\cos(\Delta t)]\\
       &\leq \frac{\Delta^2}{2}\E_{\cD}[t^2]
    \end{align*}
    where we use the fact that $1-\cos x \leq x^2/2$. Re-arranging, we find that:
    \begin{align*}
        \E_{\cD}[|t| ]&\geq  \frac{1-c\epsilon}{\Delta}\\
        \E_{\cD}[t^2 ]&\geq  \frac{2(1-c\epsilon)}{\Delta^2}
    \end{align*}
    Thus, for $\epsilon \ll 1$, $\E_{\cD}[|t|^\alpha]=\Omega(\Delta^{-\alpha})$ for $\alpha\in\{1,2\}$. The claim then follows by Jensen's inequality.
\end{proof}

Lemma~\ref{lemma:distrib-opt} admits an uncertainty-principle interpretation: to render all Fourier content at frequencies above $\Delta$ uniform, a distribution must spread over times of order $1/\Delta$, no matter how small the target error. The optimal absolute moments of $\cD$ are therefore determined by the Hamiltonian itself, rather than by the target accuracy.

As we argued in Section~\ref{sec:distrib-choice}, when $\Delta$ is unknown, we desire a distribution whose characteristic function satisfies the decay condition in Eq.~\eqref{eq:no-delta-decay}, meaning that any chosen $\beta $ benefits from quadratic error suppression. The following construction is one such distribution that satisfies this property. 

\begin{lemma}\label{lemma:gcauchy}
Let $C_\beta$ be the distribution defined by the density $f_{C_\beta}(t)=\frac 1 {\beta\pi} \frac1  {1+(t/\beta)^4/4}$. Then, $|\varphi_{C_\beta}(\omega)|\leq \sqrt{2}  e^{-\beta|\omega|}$ for any $\omega\in \R $, with $\E_{C_\beta}[|t|^\alpha]=\alpha\beta^\alpha$ for $\alpha\in \{1,2\}$ and $\E_{C_\beta}[|t|^\alpha]=\cO(\beta^\alpha)$ for $\alpha \in [1,2]$.
\end{lemma}
\begin{proof}
Re-scaling by $\beta$ without loss of generality, it suffices to show that $|\varphi_{C_1}(\omega)|\leq \sqrt 2 e^{-|\omega|}$ and $\E_{C_1}[|t|^\alpha]=\alpha$ for $\alpha\in \{1,2\}$, with the final moment claim following from Jensen's inequality.

We again begin by recalling the following standard integral identities which appear as Equations 3.241(2) and 3.727(1) in \citep{gradshteyn2014table}:
\begin{align*}
    &\int_0^\infty \frac{x^{p-1}}{1+x^q}dx =\frac \pi {q}\csc\left(\frac{p\pi}{q}\right)\quad (p,q>0)\\
    &\int_0^\infty \frac{\cos(ax)}{b^4+x^4}dx\\
    &~~~~~~~=\frac{\pi \sqrt 2 }{4b^3}\exp(-\frac{ab}{\sqrt 2})\left(\cos\left(\frac{ab}{\sqrt 2 }\right)+\sin\left(\frac{ab}{\sqrt 2}\right)\right)
\end{align*}
Applying the first for $p=1$ and $q=4$, we verify that $f_{C_1}$ is a valid probability density. Applying the second for $a=|\omega|\neq 0$ and $b=\sqrt 2 $, 
\begin{align*}
    \varphi_{C_1}(\omega)&=\frac{4}{\pi}\int_{-\infty}^{\infty}\frac{\cos(\omega t )+i\cancel{\sin(\omega t)}}{(\sqrt 2 )^4+t^4 }dt\\
    &=e^{-|\omega| }(\cos |\omega| +\sin |\omega| )\\
    &=\sqrt 2 e^{-|\omega| }\cos(|\omega|-\pi/4)
\end{align*}
where we use the fact that $\sin(\cdot),\cos(\cdot)$ are odd and even respectively. Since the $\omega=0$ case is trivial, $|\varphi_{C_1}(\omega)|\leq \sqrt 2 e^{-|\omega|}$ for any $\omega\in \R$, as desired. Finally, for the moments, by the first identity for $p=1+\alpha$ and $q=4$,
\begin{align*}
    \E_{C_1 }[|t|^\alpha ]&=\frac{2(\sqrt 2)^{\alpha+1}}{\pi}\int_0^{\infty}\frac{x^\alpha}{1+x^4}dx\\
    &=2^{(\alpha-1)/2}\csc((\alpha+1)\pi/4)\\
    &=\alpha
\end{align*}
for $\alpha\in \{1,2\}$, as desired. 
\end{proof}
We next show that the $\Theta(\beta^\alpha)$ cost of $C_\beta$ is also optimal, by an argument analogous to Lemma~\ref{lemma:distrib-opt}.

\begin{corollary}\label{cor:gcauchy-opt}
Any distribution $\cD$ satisfying Eq.~\eqref{eq:no-delta-decay}, i.e., $|\varphi_{\cD}(\omega)|\leq c e^{-\beta\omega}$ for all $\omega > 0$ and some constant $c>0$, has moment scaling $\E_{\cD}[|t|^\alpha]=\Omega(\beta^\alpha)$ for $\alpha \in [1,2]$.
\end{corollary}
\begin{proof}
Let $a\equiv \ln(2c)$ and apply Eq.~\eqref{eq:no-delta-decay} at $\omega = a/\beta$, which gives $|\varphi_{\cD}(a/\beta)|\leq c e^{-a} = \frac 1 2$. Repeating the two estimates in the proof of Lemma~\ref{lemma:distrib-opt} with $\Delta$ replaced by $a/\beta$ and $c\epsilon$ replaced by $\frac 1 2$ yields
\begin{align*}
    \E_{\cD}[|t|]\geq \frac{\beta}{2a}, \qquad \E_{\cD}[t^2]\geq \frac{\beta^2}{a^2},
\end{align*}
and the claim for $\alpha\in[1,2]$ follows by Jensen's inequality.
\end{proof}

While $S_\Delta$ and $C_\beta $ are optimal in their own respects, they are both purpose-built. We conclude by showing that the desired decay constraints can be achieved by more common distributions such as Gaussians, albeit sub-optimally in moment scaling.

\begin{lemma}\label{lemma:normal-distrib}
    Let $\mathcal{N}_{\beta,\Delta}\equiv \mathcal{N}(0,2\beta/\Delta)$. Then, $\varphi_{\mathcal N_{\beta,\Delta}}(\omega)\leq e^{-\beta\omega}$ for $\omega \geq \Delta$ and $\E_{\mathcal N_{\beta,\Delta} }[|t|^\alpha]=\Theta((\beta/\Delta)^{\alpha/2})$ for $\alpha \in [1,2]$.
\end{lemma}
\begin{proof}
Re-scaling time by $\beta$ (under which the condition $\omega \geq \Delta$ becomes $\omega \geq \beta\Delta$) and letting $A\equiv \beta \Delta$, it suffices to show that the distribution $\mathcal N_A\equiv \mathcal N (0,2/A)$ gives $\varphi_{\mathcal N_A}(\omega)\leq e^{-\omega}$ for $\omega\geq A $ with moments $\E_{\mathcal N_A}[|t|^\alpha]=\Theta(A^{-\alpha/2})$ for $\alpha \in [1,2]$.

    We begin by recalling the known characteristic function and central absolute moments of an arbitrary Gaussian:
    \begin{align*}
        \varphi_{\mathcal N (\mu,\sigma^2)}(\omega)&=e^{i\omega\mu -\sigma^2 \omega^2/2}\\
        \E_{\mathcal{N}(\mu,\sigma^2)}[|t-\mu |^\alpha]&=\frac{\sigma^\alpha 2^{\alpha/2}\Gamma((\alpha+1)/2)}{\sqrt \pi }
    \end{align*}
    where the moment formula holds for any real $\alpha > -1$.
    From the latter, $\E_{\mathcal N_A}[|t|^\alpha]=\Theta(A^{-\alpha/2})$ for $\alpha \in [1,2]$ immediately follows by taking $\sigma=\sqrt{2/A}$.
    From the former, for $\omega\geq A$,
    \begin{align*}
        \varphi_{\mathcal N _A} (\omega)&=e^{-\omega^2 / A }\leq e^{-\omega }
    \end{align*}
    as desired.
\end{proof}

Compared to $S_\Delta$, we see that $\mathcal N_{\beta,\Delta}$ trades a factor of $\Delta^{-\alpha/2}$ for an additional factor of $\beta^{\alpha/2}$ in the real-time evolution overhead. Since $\beta$ scales with $\epsilon$, we confirm that $S_\Delta$ gives asymptotically better cost than $\mathcal N _{\beta,\Delta}$, though the latter is a more well-known distribution.

\section{The Mixing Lemma as a Structure versus Randomness Statement}\label{x:svr-mixing}
Section~\ref{sec:stabilizers} developed an analogy between TITE and the structure versus randomness dichotomy of additive combinatorics~\cite{tao2007structure, gowers2001}. In this appendix, we observe that the correspondence begins one level earlier at the Campbell-Hastings mixing lemma itself, which we view as an operator-algebraic instance of what Ref.~\cite{tao2007structure} calls a \textit{generalized von Neumann theorem}. We note that its two hypotheses constitute a structure/randomness decomposition of the error ensemble, and its contrapositive is an inverse theorem. This reading unifies the abstract tool of Sec.~\ref{sec:background-rc} with the interpretation of our algorithm in Sec.~\ref{sec:stabilizers}.

To illuminate this correspondence, let our Hilbert space of interest be that of operator-valued ensembles with inner product $\langle \mathcal A,\mathcal B\rangle=\E_j {\rm tr}(A_j^\dagger B_j)$, and let our basic structured set be deterministic rank-one operators, i.e. an ensemble, $\mathcal A$, which places its entire probability mass on a single operator of the form $A=\ketbra{x}{y}$ where $\|\ket{x}\|_2=\|\ket{y}\|_2=1$. Now, consider the unitary ensemble $\mathcal U=\{U_j\}$. Then, 
\begin{equation*}
      \underbrace{\mathcal U}_{f}=\underbrace{V}_{f_{\rm str}}+\underbrace{\mathcal W}_{f_{\rm psd}}
\end{equation*}
where $\mathcal W =\{W_j\}$ with $W_j=U_j-V$ is the error ensemble. The conditions of the mixing lemma (Lemma \ref{lemma:mixing}) bound both the individual magnitudes of the errors, $\|W_j\|_{\rm op}=\|U_j - V\|_{\rm op} \leq a$, as well as the mean of the errors, $\|\E_j[W_j]\|_{\rm op}=\|\sum_j p_j U_j - V\|_{\rm op} \leq b$. Then, for any basic structured ensemble $\mathcal A $,
$$|\langle \mathcal A, \mathcal W\rangle|=|\bra{x} \E_j[W_j]\ket{y}|\leq \|\E_j[W_j]\|_{\rm op}\leq  b$$
where the final inequality is precisely the mean condition of the mixing lemma. This is exactly the statement that $\mathcal W $ is $b$-pseudorandom \citep{tao2007structure}. When $b\ll a$, the error ensemble is sufficiently pseudorandom with respect to its size; it implies that while the errors $W_j$ may be of the larger size $a$, they point in different directions such that their average correlates with any fixed operator by at most $b$. 

The conclusion of the lemma is then the generalized von Neumann step of the dichotomy. The channel $\rho \mapsto U\rho U^\dagger$ is a \emph{quadratic} functional of $U$. Expanding and grouping the terms by moment order, 
\begin{equation*}
    \Lambda(\rho) - \mathcal V(\rho) = \E_j[W_j]\,\rho V^\dagger + V\rho\, \E_j[W_j]^\dagger + \E_j[W_j \rho W_j^\dagger] , 
\end{equation*}
meaning the linear terms contribute at most $2b$ by the $b$-pseudorandomness of $\mathcal W$, while the quadratic term must be naively controlled by the size of each error, contributing $a^2$. This is the central mechanism of the paradigm, namely that the pseudorandomness of the error ensemble is what suppresses the first moment contributions to the diamond norm error of the channel. This is morally equivalent to the generalized von Neumann results described in Ref.~\citep{tao2007structure} which argue that pseudorandom objects have poor multilinear correlations. In this light, the contrapositive of the mixing lemma is an inverse theorem; if the mixed channel misses its target by much more than $a^2$, then necessarily $b \gg a^2$, meaning the error ensemble is not very pseudorandom and instead is strongly correlated with a fixed operator direction.


Finally, we note an instructive inversion of purpose. Additive combinatorics decomposes a given object into structure plus randomness in order to prove theorems about arbitrary objects. The mixing lemma runs this in reverse; given a structured target $V$ that no single accessible unitary can reach, it seeks to construct an object $\mathcal{U}$ whose errors are quadratically pseudorandom with respect to their individual size.

\bibliography{refs}

\end{document}